\RequirePackage{amsmath}
\documentclass[a4paper,runningheads]{llncs}

\usepackage[T1]{fontenc}
\usepackage{graphicx}
\usepackage{multirow}
\usepackage{amssymb}
\usepackage{mathtools}
\usepackage{todonotes}
\usepackage[basic,small,classfont=bold]{complexity}
\usepackage[capitalise]{cleveref}
\usepackage{proof}
\usepackage{lineno}
\usepackage{algorithm}
\usepackage[noend]{algpseudocode}
\usepackage{nicefrac}
\usepackage{textcomp}
\usepackage{xspace}
\usepackage{subcaption}
\usepackage{tikz}
\usetikzlibrary {datavisualization}
\usetikzlibrary{datavisualization.formats.functions,shadows,calc}
\usepackage{tikz}
\usetikzlibrary{shapes,shapes.geometric,arrows,fit,calc,positioning,automata,chains,matrix.skeleton, arrows.meta}
\newcommand{\pgfmathparseFPU}[1]{\begingroup%
\pgfkeys{/pgf/fpu,/pgf/fpu/output format=fixed}%
\pgfmathparse{#1}%
\pgfmathsmuggle\pgfmathresult\endgroup}

\newcommand{\ie}{i.e.~}

\newcommand{\Set}[1]{\{#1\}}
\newcommand{\close}[1]{\ensuremath{#1 \mathrm{\downarrow}}}
\newcommand{\maxac}[1]{\ensuremath{\lceil #1 \rceil}}
\newcommand{\Reg}{\mathrm{Reg}}
\newcommand{\bwd}{\mathrm{bwd}}

\newcommand{\cpre}{\mathrm{CPre}}
\makeatletter
\def\finish@kdtree{\ifx\next-\relax%
    -tree%
  \else%
    ~tree%
  \fi\xspace}
\newcommand{\kdtree}{\(k\)-d\futurelet\next\finish@kdtree}
\makeatother

\newcommand{\kdtrees}{\(k\)-d~trees\xspace}

\title{Data Structures for Finite Downsets of Natural Vectors: Theory and
  Practice}
\titlerunning{Data Structures for Finite Downsets of Natural Vectors}
\author{Micha\"el Cadilhac\inst{1}
\and
Vanessa Fl\"ugel\inst{2} \and
Guillermo A. P\'erez\inst{2}
\and
Shrisha Rao\inst{2}
}
\authorrunning{M. Cadilhac et al.}
\institute{DePaul University, Chicago, USA\\
\email{michael@cadilhac.name} \and
University of Antwerp -- Flanders Make, Antwerp, Belgium\\
\email{\{guillermo.perez,shrisha.rao\}@uantwerpen.be}}

\tikzset{
  data visualization/our system/.style={
    new Cartesian axis=x axis,
    new Cartesian axis=left axis,
    new Cartesian axis=right axis,
    x axis={attribute=x},
    left axis={attribute=y,unit vector={(0cm,1pt)}},
    right axis={attribute=r,unit vector={(0cm,1pt)}},
    left axis={
        visualize label={
          x axis={goto=padded min},
          node style={
            at={(0,3.2cm)},
            above
          }
        },
        label=sec.},
      right axis={
        visualize label={
          x axis={goto=padded max},
          node style={
            at={(0,3.2cm)},
            above
          }
        },
        label=ratio},
      x axis={ visualize label={
          x axis={goto pos=.5},
          node style={
            at={(0,0 |- data visualization bounding box.south)},
            below
          }
        },
        label=antichain size},
    x axis    ={length=\pgfkeysvalueof{/tikz/data visualization/scientific axes/width}},
left axis ={length=\pgfkeysvalueof{/tikz/data visualization/scientific axes/height}},
right axis={length=\pgfkeysvalueof{/tikz/data visualization/scientific
    axes/height}},
    every axis/.style={style=black!50}, 
    left axis= {visualize axis={x axis=   {goto=min}}},
    right axis={visualize axis={x axis=   {goto=max}}},
    x axis= {visualize axis={left axis={goto=min}},
      visualize axis={left axis={goto=max}}},
     all axes= {padding=.5em},
     left axis= {visualize axis={x axis=   {goto=padded min}}},
     right axis={visualize axis={x axis=   {goto=padded max}, style=blue}},
    x axis= {visualize axis={left axis={goto=padded min}}},
    left axis={visualize ticks={style={red!50!black}, direction axis=x axis,
        x axis={goto=padded min}, high=0pt, tick text at low}},
    right axis={visualize ticks={style={blue!80!black}, direction axis=x axis,
        x axis={goto=padded max}, low=0pt, tick text at high}},
    x axis={visualize ticks={direction axis=left axis, left axis={goto=padded min}, high=0pt,
        tick text at low},
    },
    left axis={include value={0}},
    right axis={include value={0}},
    x axis={include value={0}},
    all axes={ticks},
  }
}

\begin{document}

\maketitle

\begin{abstract}
  Manipulating downward-closed sets of vectors forms the basis of so-called
  antichain-based algorithms in verification.  In that context, the dimension of
  the vectors is intimately tied to the size of the input structure to be
  verified.  In this work, we formally analyze the complexity of classical
  list-based algorithms to manipulate antichains as well as that of Zampuni\'eris's sharing trees and traditional
  and novel \kdtree-based antichain algorithms.  In contrast to the existing
  literature, and to better address the needs of formal verification, our
  analysis of \kdtree algorithms does not assume that the dimension of the vectors is fixed.
  Our theoretical results show that \kdtrees are asymptotically better than both
  list- and sharing-tree-based algorithms, as an antichain data structure, when the antichains
  become exponentially larger than the dimension of the vectors.  We evaluate
  this on applications in the synthesis of reactive systems from linear-temporal
  logic and parity-objective specifications, and establish empirically that
  current benchmarks for these computational tasks do not lead to a favorable
  situation for current implementations of \kdtrees.

  \keywords{Antichain algorithms
  \and Data structures}
\end{abstract}

\section{Introduction}
The efficiency and scalability of 
verification techniques 
such as
model checking and temporal synthesis largely depend on the size and
complexity of the input. One way to mitigate this consists in making
use of (symbolic) data structures to represent the model implicitly~\cite{DBLP:reference/mc/2018}. For instance, when the states of a model admit
some partial order and sets of states satisfying properties of interest are
downward-closed with respect to the partial order, verification algorithms can
manipulate sets of states 
by storing their (antichain of) maximal
states only. 

Antichain-based algorithms exist for various verification problems. There are antichain-based algorithms for stable failures refinement
checking, for failures-divergence refinement checking, and for probabilistic
refinement checking 
~\cite{DBLP:conf/icfem/WangS0LDWL12,DBLP:journals/lmcs/LaveauxGW21}.
Antichains have also been used in algorithms to check the inclusion between the
sets of traces (\ie data languages) recognized by generic register automata---this is undecidable, but there is a semi-algorithm based on  abstraction refinement and antichains that is known to be sound
and complete~\cite{DBLP:journals/fmsd/HolikIRV20}.  There are also antichain
algorithms to solve universality and language inclusion problems for
nondeterministic B\"uchi automata, and the emptiness problem for alternating
B\"uchi automata~\cite{DBLP:journals/corr/abs-0902-3958}. As a final example, there are antichains algorithms for the inclusion problem between
infinite-word visibly pushdown languages~\cite{DBLP:conf/tacas/DoveriGH23}.

It is important to note that there is no easy template for antichain-based
algorithms. Each algorithm mentioned above exploits a slightly different partial
order and must argue that different operations preserve closure of the sets to
be able to manipulate antichains only. In this work, we are primarily interested
in state spaces that can be encoded as vectors of natural numbers in $\mathbb{N}^k$. The partial
order we consider is the product order, that is, the component-wise order. There
are a number of antichain-based algorithms where such an encoding into natural
vectors arises naturally. They appear, for example, in checking safety properties for Petri net markings, such as mutual exclusion and coverability~\cite{DBLP:journals/sttt/DelzannoRB04} and in antichain algorithms to solve parity
and mean-payoff games with imperfect
information~\cite{DBLP:books/cu/11/0001R11,DBLP:journals/iandc/BerwangerCWDH10,DBLP:journals/acta/HunterPR18}, where
the vectors keep track of visits to vertices with ``bad priorities'' and running
sums of weights, respectively. Finally, in antichain algorithms for the synthesis
of reactive
systems~\cite{DBLP:journals/fmsd/FiliotJR11,DBLP:conf/tacas/CadilhacP23}, the
vectors keep track of visits to rejecting states (in universal co-B\"uchi
automata).

Despite the abundance of antichain-based algorithms, 
there seems to be no systematic study of the complexity of
antichain-manipulating operations. There are two notable exceptions: First,
there are algorithms and 
upper bounds for the computation of the
antichain of maximal elements given a downward-closed
set~\cite[Sec. 4]{DBLP:journals/siamcomp/DaskalakisKMRV11}. 
Second, 
a result from~\cite{DBLP:journals/fmsd/FiliotJR11} implies that computing
the intersection of a given set of antichains cannot be done in polynomial time
unless $\NP = \P$.

Similarly, the body of work focusing on data structures to support the
manipulation of antichains is surprisingly slim. Delzanno et al.~\cite{DBLP:journals/sttt/DelzannoRB04,gantySymbolicDataStructure2007} propose a representation of downward-closed sets based on sharing trees, an extension of binary decision diagrams (BDDs) for natural vectors proposed by  Zampuni\'eris~\cite{zampunieris1997sharing}. However, the work of Zampuni\'eris focuses on general sets of vectors instead of downward-closed sets and the works by Delzanno et al. do not reduce downward-closed sets to antichains of maximal elements but instead use approximate reductions.
 In~\cite{DBLP:conf/tacas/CadilhacP23}, it is suggested that some version of \kdtrees may work well for closed sets of natural vectors.  These are classic data structures used in computational
geometry to implement efficient range-search algorithms~\cite{kdtree_book}. 
In
particular, they support fast dominance queries, which 
can be used
to implement 
membership of a vector in a downward-closed set. 
Unfortunately,
in~\cite{DBLP:conf/tacas/CadilhacP23} it is also reported that a simple textbook
implementation of \kdtrees does not outperform list-based versions of antichain
operations.  We postulate that this is due to the fact that the
algorithm analyses and optimizations for \kdtrees described in the computational
geometry literature make the assumption that the dimension $k$ of the vectors is
fixed or at least logarithmically small with respect to the rest of the input~\cite{DBLP:journals/dcg/Chan19}.  In contrast, for antichain-based algorithms, the dimension is a crucial
part of the input.  
For instance, in the reactive-system synthesis application, $k$ corresponds to the number of states of a (co-)B\"uchi automaton, which is typically large, while small sets of vectors with small entries are not uncommon.

\paragraph{Our contributions.}
In this work, we formally compare the worst-case running times of classic
list-based, novel (antichain-ensuring) sharing-tree-based, and novel \kdtree-based antichain algorithms. On the way, we
re-examine classic tree-building and searching algorithms for \kdtrees. These lead to a more efficient implementation of \kdtrees than the simple textbook one used in~\cite{DBLP:conf/tacas/CadilhacP23}.  We
have implemented optimized versions of list-, sharing-tree-, and \kdtree-based algorithms
in a library written in C++20, enabling an empirical evaluation of the
algorithms.
This evaluation is carried out on
applications to synthesis of reactive systems inside
Acacia-Bonsai~\cite{DBLP:conf/tacas/CadilhacP23} and Oink~\cite{vanDijk18a,knor}.

\section{Preliminaries}
We write $\mathbb{N}$ for the set of all nonnegative integers, including $0$. When we
need to exclude $0$, we instead write $\mathbb{N}_{>0}$.
Let $k \in \mathbb{N}_{>0}$ and consider vectors
$\vec{u},\vec{v} \in \mathbb{N}^k$. Then, $\vec{u}$ is \emph{smaller} than
$\vec{v}$ (equivalently, $\vec{v}$ is \emph{larger} than $\vec{u}$), denoted by
$\vec{u} \leq \vec{v}$, if each component of $\vec{v}$ is greater than or equal to
that component of $\vec{u}$ and the inequality is strict if the relation is
strict for at least one component. If $\vec{u}$ is neither larger, nor smaller
than $\vec{v}$, then they are said to be \emph{incomparable}.

\begin{definition}[Downset]
    Let $k \in \mathbb{N}_{>0}$. A set $V \subset \mathbb{N}^k$ is a \emph{downset} if all vectors smaller than any vector in the set are also in the set, \ie if for all $\vec{u} \in \mathbb{N}^k$ and $\vec{v} \in V$ we have that $\vec{u} \leq \vec{v}$ implies $\vec{u} \in V$.
\end{definition}

For any $V \subset \mathbb{N}^k$, we use $\close{V}$ to denote the \emph{downward closure of $V$}, \ie the set obtained by adding to $V$ all the vectors smaller than the vectors in $V$. In symbols, $\close{V} = \Set{\vec{u} \in \mathbb{N}^k \mid \exists \vec{v} \in V, \vec{u} \leq \vec{v}}$. Note that the downward closure of any set of vectors is a downset.

Next, we define \emph{antichains} of natural vectors. These arise, for instance, when considering the maximal elements of a downset.
\begin{definition}[Antichain]
    Let $k \in \mathbb{N}_{>0}$. A set $V \subset \mathbb{N}^k$ is an \emph{antichain} if it consists only of pairwise incomparable vectors, \ie for all $\vec{u},\vec{v} \in V$ we have neither $\vec{u} \leq \vec{v}$ nor $\vec{u} \leq \vec{v}$ (\ie they are incomparable).
\end{definition}

\noindent
It follows from Dickson's lemma that any antichain (of natural vectors) is finite.

For a finite set $V \subset \mathbb{N}^k$, we write $\maxac{V}$ to denote the \emph{antichain of maximal elements from $V$}. In symbols, $\maxac{V} = \Set{\vec{v} \in V \mid \forall \vec{u} \in V, \vec{v} \not< \vec{u}}$.

\paragraph{Antichain algorithms.}
We are interested in efficient algorithms to manipulate $k$-dimensional downsets, for $k \in \mathbb{N}_{>0}$. Note that for any finite downset $V \subset \mathbb{N}^k$, we have that $V = \close{\maxac{V}}$. This follows from the result below, which is itself a simple consequence of the transitivity of the partial order on vectors.
\begin{proposition}\label{pro:trans}
    Let $k \in \mathbb{N}_{>0}$, $u \in \mathbb{N}^k$, and a finite set $V \subset \mathbb{N}^k$. Then, $\vec{u} \in \close{V}$ if and only if $\vec{u} \leq \vec{v}$ for some $\vec{v} \in \maxac{V}$.
\end{proposition}

\noindent
Hence, one can store the antichain $A = \maxac{V}$ as a representation of $V = \close{A}$.

To realize (binary) operations on downsets represented by the antichains $A$ and $B$, an algorithm can avoid (as much as possible) to explicitly compute the downward closures of the antichains. Further, the result of the operation should also be output by the algorithm represented as antichains. Since, we focus on union and intersection, the following result implies that a single antichain suffices as output of such an algorithm.
\begin{proposition}
    Let $k \in \mathbb{N}_{>0}$ and $U,V \subset \mathbb{N}^k$. If $U$ and $V$ are both downsets, then $U \cup V$ and $U \cap V$ are also downsets.
\end{proposition}

\noindent
Concretely then, given $A = \lceil U \rceil$ and $B = \lceil V \rceil$, we are looking for efficient algorithms to compute $\maxac{\close{A} \cup \close{B}}$ and $\maxac{\close{A} \cap \close{B}}$. In the rest of this paper we consider data structures to encode $A$, $B$, and 
their union or intersection.

Finally, there is a decision problem whose complexity is related to that of (all variations) of the union and intersection algorithms present in this paper. Namely, given a $k$-dimensional vector $\vec{v}$ and an antichain $A \subset \mathbb{N}^k$ representing a downset, determine whether $\vec{v} \in \close{A}$. We call this the \emph{membership problem}.

\paragraph{Algorithm analysis.}
For all our analyses, our working computational model is that
of a random-access machine instead of a Turing machine. In particular, this
means we assume indirect addressing takes constant time. Furthermore, we
assume comparing two natural numbers $a,b \in \mathbb{N}$
takes constant time. That is, all of the following 
constitute 
atomic operations: $a \leq b$, $a < b$, and $a = b$.

\section{List-based antichain algorithms}
Fix $k \in \mathbb{N}_{>0}$ and a finite $k$-dimensional downset represented by the antichain $A \subset \mathbb{N}^k$ of its maximal elements. Write $m$ for the cardinality of $A$, \ie $|A| = m$.

\subsection{The membership problem}
The na\"{i}ve approach to check whether a vector $\vec{v} \in \mathbb{N}^k$ is in the
downset $\close{A}$ consists in comparing it, component by component, with each
vector in $A$.  Correctness of this approach follows from \Cref{pro:trans}.
This requires $k m$ time in the worst case since checking whether a vector is
smaller than another vector can be done in $k$ comparisons of numbers.
\begin{proposition}\label{pro:list-member}
    List-based algorithms solve membership in time $O(km)$.
\end{proposition}

\subsection{The union operation}\label{sec:list-union}
Henceforth, let us fix a second antichain $B \subset \mathbb{N}^k$ representing a $k$-dimensional downset. Further, write $|B| = n$ for the cardinality of $B$.

To compute the antichain corresponding to $\maxac{\close{A}\cup \close{B}}$, we first note that this antichain will be a subset of vectors in $A$ and $B$. 
\begin{lemma}[{\cite[Prop. 2, (i)]{DBLP:journals/fmsd/FiliotJR11}}]\label{lem:union}
    Let $C = \maxac{\close{A} \cup \close{B}}$. Then, $C \subseteq A \cup B$. Moreover, for all $\vec{u} \in A$, we have $\vec{u} \not\in C$ if and only if $\vec{u} < \vec{v}$ for some $\vec{v} \in B$.
\end{lemma}

Based on this, we can check, for each vector in $A$, whether it is in the downward closure of $B$ using the membership algorithm described earlier. If the answer is yes, the vector from $A$ is discarded, if not, it is kept for the resulting antichain. We repeat the process once again but now checking, for each vector in $B$, whether some vector in the downward closure of $A$ is strictly larger than it. (This can be done with a minimally modified membership algorithm.) Both of these checks take $kmn$ time in the worst case: when the new antichain is $A\cup B$.
\begin{proposition}
    List-based algorithms for union run in time $O(kmn)$.
\end{proposition}

We describe a small optimization for the union algorithm that halves the number of comparisons required. Instead of comparing each element of $A$ with each element of $B$ twice, when comparing the components of a vector $\vec{u} \in A$ with those of $\vec{v} \in B$, we check for equality until we find a dimension $1\leq i\leq k$ such that $u_i \neq v_i$. Let us assume, without loss of generality, $u_i > v_i$. We can now conclude that $\vec{u} \not< \vec{v}$. Thus, for the remaining components $i<p\leq k$, we only need to check whether $u_p>v_p$, to determine whether $\vec{u} > \vec{v}$. In total, these are $k+1$ comparisons for each pair of vectors instead of $2k$.

\subsection{The intersection operation}\label{sec:list-inter}
For the intersection, we note that $\maxac{\close{A} \cap \close{B}}$ will be a subset of the \emph{meets} of vectors in $A$ and $B$, \ie their component-wise minimum. Formally, the meet of two vectors $\vec{u}, \vec{v} \in \mathbb{N}^k$ is defined as 
\(
    \vec{u} \sqcap \vec{v} =(\min(u_1,v_1), \dots,\min(u_k,v_k)).
\)
The notion is extended to sets: 
$U \sqcap V$ of $U,V \subset \mathbb{N}^k$ is the set $\Set{\vec{u} \sqcap \vec{v} \mid \vec{u} \in U, \vec{v} \in V}$.
\begin{lemma}[{\cite[Prop. 2, (ii)]{DBLP:journals/fmsd/FiliotJR11}}]\label{lem:inter}
    Let $C = \maxac{\close{A} \cap \close{B}}$. Then, $C = \maxac{A \sqcap B}$.
\end{lemma}
The cardinality of $A \sqcap B$ is at most $mn$. Now, to compute its antichain of maximal elements, we check the membership of each vector in $A \sqcap B$ within the downward closure of the other vectors using the membership algorithm introduced earlier. This results in a quadratic worst-case running time.
\begin{proposition}\label{pro:inter-list}
    List-based algorithms for 
    intersection 
    run in time $O(km^2n^2)$.
\end{proposition}

Observe that intersection is much more costly than union even when working
with antichains as a representation of downsets. 
There is no known approach to computing intersection which avoids
the explicit computation of the set of meets of both antichains. We do know that the intersection of a given set of antichains cannot be done in polynomial time unless $\P = \NP$~{\cite[Proposition 5]{DBLP:journals/fmsd/FiliotJR11}}.
%

As a small optimization, we
observe that a vector $\vec{v} \in A$ can be excluded from the computation of
all meets if the membership check $\vec{v} \in \close{B}$ is positive, \ie
there is some $\vec{w} \in B$ such that $\vec{v} \leq \vec{w}$. This is
because all elements in $\{\vec{u} \sqcap \vec{v} \mid \vec{u} \in B\}$ will
be smaller than $\vec{v} \sqcap \vec{w} = \vec{v}$. This can be
used to reduce the number of meets that must be considered. 
To conclude, one could also follow the analysis from~\cite[Theorem 11]{DBLP:journals/siamcomp/DaskalakisKMRV11} to establish a bound which replaces one of the $mn$ factors in \cref{pro:inter-list} by the \emph{width}, \ie the largest antichain over $A \sqcap B$.

\section{Sharing-tree-based antichain algorithms}
As before, let us fix $k \in \mathbb{N}_{>0}$ and a finite $k$-dimensional downset represented by the antichain $A \subset \mathbb{N}^k$ of its $m$ maximal elements. An additional parameter of interest is $W = \max_{\vec{v} \in A} \lVert \vec{v} \rVert_\infty$, the maximal norm over elements in $A$. 

In this section, instead of encoding $A$ as a plain list, we construct a \emph{sharing tree} \cite{zampunieris1997sharing} for it. 
One can look at each vector $v\in V$ as a word over $\mathbb{N}$ of length $k$. Hence, $A$ is a (finite) regular language. The sharing tree of $A$ is nothing more than its minimal deterministic \emph{acyclic} finite-state automaton (DFA) \cite{Daciuk,Blumer}.

\begin{definition}[Sharing tree]
\label{def:st}
A sharing tree\footnote{It is a directed acyclic graph, not a tree, but it was thus named in earlier literature.} $(N, r, \mathrm{val}, \mathrm{succ})$ is a rooted acyclic graph  where \(N = N_0 \uplus \dots \uplus N_{k}\) is a set of nodes, partitioned into $k+1$  \emph{layers}, \(N_0 = \{r\}\) is the root, \(\mathrm{val} \colon N \rightarrow \mathbb{N}\cup\{\top\}\) is a value-labeling function that satisfies $\mathrm{val}(r)=\top$, and \(\mathrm{succ}\colon N \rightarrow 2^N\) is the successor function. Additionally:
\begin{enumerate}
    \item For all $0 \leq i < k$ and $n \in N_i$, $\mathrm{succ}(n) \subseteq N_{i+1}$, i.e., nodes can only have edges to the next layer;
    \item For all $n \in N$ and $s_1\neq s_2 \in \mathrm{succ}(n)$, it holds that \(\mathrm{val}(s_1) \neq \mathrm{val}(s_2)\), i.e., two successors cannot have the same value;
    \item For all $0 \leq i < k$ and $n_1\neq n_2 \in N_i$, \(\mathrm{val}(n_1) = \mathrm{val}(n_2) \implies \mathrm{succ}(n_1) \neq \mathrm{succ}(n_2)\), i.e., same-layer nodes cannot have the same label and successors;
    \item For all $n \in N_{k}$, \(succ(n) = \emptyset\), i.e.,  nodes on the last layer have no successors;\label{itm:def3-3}
    \item If $n\in N_1\cup\dots\cup N_{k-1}$, then $\mathrm{succ}(n)\neq\emptyset$, that is, each node which is not $r$ or in $N_{k}$ must have at least one successor.
\end{enumerate}
\end{definition}

Note that we label nodes of the graph and not the edges for ease of notation. Equivalently, one can label all the incoming edges of a node by the value of that node to obtain a DFA in the usual sense.

\subsection{Growing sharing trees}
Our algorithm for building a sharing tree is given in \Cref{alg:build_stree}, where $j = 0$ initially and (for technical convenience) $v_0 = \top$ for all $\vec{v} \in A$. Intuitively, one builds an (implicit) trie for $A$ and then minimizes the resulting DFA in a bottom-up fashion. Minimization, \ie merging of the (language-)equivalent nodes, is then handled \`a la Revuz \cite{DBLP:journals/tcs/Revuz92} by giving unique identifiers to the nodes based on their label and the indices of their successors in their layer.\footnote{By \Cref{itm:def3-3} this combination is indeed unique per layer.}
\begin{algorithm}
\caption{BuildSharingTree($A,j$)}\label{alg:build_stree}
\begin{algorithmic}[1]
\If{$j=k$}
    \State Return (cached) leaf $T$ with label $v_k$, where $A = \{\vec{v}\}$
\EndIf
\ForAll{$\vec{v} \in A$}
    \State add $\vec{v}$ to bucket[$v_{j+1}$]
\EndFor
\State $T \gets$ tree node with label $v_j$ 
\ForAll{$a \in \mathbb{N}$ such that bucket[$a$] is nonempty, in decreasing order}
    \State $T_a, \mathrm{idx}_{j+1}(a) \gets$ BuildSharingTree(bucket[$a$], $j+1$) \Comment{New node \& its index}
    \State add $T_a$ as successor of $T$
\EndFor
\State Set id of $T$ to $(j,v_j,\mathrm{idx}_{j+1}(0), \mathrm{idx}_{j+1}(1)\dots)$ \Comment{Uniquely identifies this subtree}
\State Return (cached) (sub)tree root $T$
\end{algorithmic}
\end{algorithm}

Note that every layer of the trie will have at most $m$ nodes. Then, since we have $k$ layers, and every recursive call of the algorithm runs in $O(W)$ time (initializing the buckets to empty already costs that much), we get that the tree can be computed in time $O(kmW)$. Observe that if $W > m$ one can ``compress'' the vectors by sorting them per dimension (in time $O(km \log m)$) and keeping only their position in the corresponding sorted list. This ensures $W \leq m$ and one can store the sorted lists as an $m \times k$ matrix to keep the original values (the components of which can be accessed in constant time via the stored indices).

\begin{lemma}\label{lem:build-st}
    A sharing tree, with $O(km)$ nodes and edges, for a set of $m$ vectors in dimension $k$ with max norm $W$ can be constructed in time $O(km\min(m,W))$.
\end{lemma}

To avoid wasting memory due to sparse successor arrays, the trie can be built by first sorting the vectors based on the relevant component and using a linked list instead of an array with $\{0,1,\dots,W\}$ as indices. This also means we can conveniently store all successors in decreasing order of their labels. Additionally, instead of an exact cache for nodes, one can use a hash table.

Although building a sharing tree for $A$ is more costly than storing it as a list, it may result in exponential savings in terms of space and subsequent membership queries. Indeed, due to sharing, the constructed sharing tree could be exponentially smaller than the starting set. For example, the antichain represented by the language $\{01+10\}^n$ has size exponential in $n$, but can be represented as a sharing tree of size $O(n)$ (cf.~\cite[Sec. 1.3.3]{zampunieris1997sharing} and~\cite[Prop. 2]{DBLP:journals/sttt/DelzannoRB04}).

\subsection{The membership problem}

To determine whether a vector $\vec{u} \in \mathbb{N}^k$ is part of the downset $\close{A}$, we can use a depth-first search (DFS). Starting from the root, the values of the successors are compared to the corresponding component of $\vec{u}$ at every layer. If the value of the node is greater than or equal to the component, the branch is followed further down the sharing tree or membership is confirmed, in case of being at the last layer, otherwise it is discarded. Based on the linear-time complexity of DFS and \Cref{lem:build-st}, we get the following result.

\begin{proposition}
    Sharing-tree-based algorithms solve membership in time $O(km)$.
\end{proposition}

\noindent
We again remark that the sharing tree may be of logarithmic size with respect to $A$. Hence, the DFS may in fact be much more efficient than a search over $A$, even if the conclusion is that the given vector is not a member of $\close{A}$.

As a small optimization, we note that having stored the successors of every node in decreasing label order allows for early exits. Indeed, if during the DFS we encounter a node whose first successor has a smaller label than the corresponding component of $\vec{u}$, there is no need to check the remaining branches following the remaining successors. Similarly, when searching at the second-to-last layer of the sharing tree, a dominating vector is encoded by some branch only if the label of the last successor is larger than the last component of $\vec{u}$.

\subsubsection{Covering sharing trees.} We are using Zampuni\'eris' version of sharing trees~\cite{zampunieris1997sharing} while ensuring we encode an antichain. This is in contrast to \emph{covering sharing trees}, as proposed in \cite{DBLP:journals/sttt/DelzannoRB04}, which do not encode the antichain of maximal elements only. Due to the usage of an \emph{approximate} domination check, covering sharing trees could be much larger than the sharing tree of the antichain of maximal elements. Moreover, the union and intersection algorithms for covering sharing trees are more graph-based and reminiscent of BDD operations than the ones we give below (see following subsection, cf. \Cref{sec:csts}). These facts make complexity comparisons with the other approaches in this work more difficult and less interesting. In \Cref{sec:experiments} we do present some empirical evidence showing that ensuring that the encoded set is an antichain results in size and time gains.

\subsection{The union and intersection operations}
As before, for the binary operations, we shall fix a second antichain $B \subset
\mathbb{N}^k$ with $n$ vectors representing a $k$-dimensional downset. We also extend our bound on the maximal norm so that $W = \max_{\vec{v} \in A \cup B} \lVert \vec{v} \rVert_\infty$. Our sharing-tree-based algorithms for the binary operations are mostly identical to the list-based ones. However, their complexity is higher due to the added cost of building sharing trees compared to just keeping a list of vectors.

Recall that the antichain $C = \maxac{\close{A} \cup
\close{B}}$ is a subset of both $A$ and $B$, see~\Cref{lem:union}. Moreover, for both $A$ and $B$, the vectors are added to $C$ if they are not in the closure of the respective other set, so if there is no strictly larger vector. As shown, these checks can be realized using a DFS in time
$O(kmn)$. Finally, constructing the sharing tree for $C$ takes time
$O(k(m+n)\min(m+n, W))$.
\begin{theorem}
    There is a sharing-tree-based algorithm for the union operation that runs in time 
    $O(kn(m+\min(n, W))$ assuming (w.l.o.g.) that $m \leq n$.
\end{theorem}

For intersection, our starting point is the set $A \sqcap B$, as per \Cref{lem:inter}.
To compute the antichain $C = \maxac{A \sqcap B}$, we first construct a sharing tree for the set $A \sqcap B$. Then, we use the sharing tree to check, for each vector in the set, whether it is strictly smaller than some other vector. Finally, we construct a second sharing tree for the vectors for which the result of the previous check was negative (\ie the maximal elements from $A \sqcap B$). The number of possible meets $A \sqcap B$ is at most $mn$ and can be enumerated in time $O(kmn)$. Conveniently, constructing the sharing tree removes duplicates due to the minimization step. So constructing a sharing tree for $A \sqcap B$ is as simple as enumerating all the meets and constructing a sharing tree, the latter in time $O(kmn\min(mn, W))$. Hence, building the first sharing tree can be done in time $O(kmn\min(mn, W))$.
Then, for each vector in the set, the strict membership takes time $O(kmn)$, yielding a total time of $O(k(mn)^{2})$. To conclude, we note that, in the worst case, constructing the second tree is just as costly as constructing the first one.

\begin{theorem}
    There is a sharing-tree-based algorithm for the intersection operation that runs in time $O(km^2n^2)$.
\end{theorem}

\section{\kdtree-based antichain algorithms}

One final time, let us fix $k \in \mathbb{N}_{>0}$ and a finite $k$-dimensional downset represented by the antichain $A \subset \mathbb{N}^k$ of its $m$ maximal elements. In this section, we encode $A$ using a (real) tree structure that generalizes binary search trees to higher dimensions. In fact, we will make use of \kdtrees.

We closely follow the presentation from~\cite[Ch. 5.2]{kdtree_book}. To present
the theory as close as possible to what we implement for the forthcoming
experiments, we introduce some additional notation. Mainly, we need a clear
definition of the median of a sorted list with repeated elements and a total
strict order on the elements of the list (cf.~\cite[Ch. 5.5]{kdtree_book}).

Consider a sequence $x_1, x_2, \dots, x_p$ of natural numbers. Let us write $x_i \prec x_j$ if and only if $x_i < x_j$ or $x_i = x_j$ and $i < j$. That is, $\prec$ is the lexicographic order obtained from the order on the elements in the sequence and their indices. Now, the $\prec$-median of the sequence is the $\lceil \nicefrac{p}{2} \rceil$-th $\prec$-largest number in the sequence.


\subsection{Growing \kdtrees}\label{sec:build-kdtree}
Intuitively, a \kdtree is the natural generalization of a binary search tree from $1$ to multiple dimensions. When going down a branch of the tree by $j$ levels, the remaining set of vectors is partitioned based on the $(j+1)$-th coordinate. The algorithm to build a \kdtree is given in \Cref{alg:build_tree}, where $j=0$ initially.

\begin{algorithm}
\caption{BuildKDTree($A,j$)}\label{alg:build_tree}
\begin{algorithmic}[1]
\If{$A=\Set{\vec{v}}$}
    \State Return leaf storing $\vec{v}$
\EndIf
    \State $i \gets (j \mod k) + 1$ \label{code:mod}\Comment{The current dimension}
    \State $\vec{\mu} \gets$ the $\prec$-median of $A$ with respect to the $i$-th coordinate
    \State $B_{\prec} \gets \Set{\vec{v} \in A \mid v_i \prec \mu_i}$
    \State $B_{\succ} \gets \Set{\vec{v} \in A \mid v_i \succ \mu_i} \cup \Set{\vec{\mu}}$
    \State $T_{\prec} \gets$ BuildKDTree($B_\prec,j+1$)
    \State $T_{\succ} \gets$ BuildKDTree($B_{\succ},j+1$)
    \State Return tree $T$ with children $T_\prec$ and $T_{\succ}$ and label $\mu_i$
\end{algorithmic}
\end{algorithm}

\noindent
Note that when $k$ is small, e.g. $2$, then we need to make sure we are still partitioning the vectors according to valid dimensions even after the second layer of the tree. This is why \Cref{code:mod} requires a modulo operator.

It is instructive to consider the case where all elements of $A$ are unique with respect to all coordinates. Here, one can think of each step of the algorithm as computing the median of the vectors according to dimension $i$, splitting the space $\mathbb{N}^k$ into two on the $i$-th coordinate: the left subtree encodes the region of space consisting of all the vectors whose $i$-component is strictly less than that of the median; the right one, all those greater than or equal to it. The use of $\prec$ is really only a technicality useful to ensure the tree is balanced.

\begin{lemma}[{\cite[Lemma 5.3]{kdtree_book}}]\label{lem:build-complexity}
    A \kdtree for a set of $m$ vectors in $k$ dimensions can be constructed in time $O(m \log m)$.
\end{lemma}

We highlight a discrepancy between \kdtrees in fixed dimension and those where $k$ is part of the input. In most presentations of \kdtrees (e.g.~\cite{kdtree_book,brass_2008}) it is suggested that computing the median on every recursive call of the algorithm to build the tree can be avoided by doing some preprocessing: compute a sorted version of the initial list for each dimension and  split these into sorted sublists for the recursive calls. Note that even the presorting introduces a dependency on $k$ as sorting $k$ lists costs us $O(km \log m)$ if the dimension is part of the input. Hence, computing the median in each recursive call does seem better in our case.

\subsection{The membership problem}\label{sec:ac-membership}
Let $\vec{u} \in \mathbb{N}^k$. Now, given a \kdtree $T$ encoding the antichain $A$, we will appeal to \Cref{pro:trans} and use $T$ to determine whether $\vec{u} \in \close{A}$ by searching the tree to find some $\vec{v} \in A$ such that $\vec{u} \leq \vec{v}$. First, we introduce some additional notation.

To each internal node $T$ we associate a region\footnote{Technically, we are dealing with the natural vectors contained in the region. We find that the name region still conveys the right intuition.} $\Reg(T) \subseteq \mathbb{N}^k$ inductively. We start with the entirety of $\mathbb{N}^k$ at the root. Then, for a node $T$ at depth $j$ and with $\vec{\mu}$ the $\prec$-median used to build it, to its right child $T_\succ$, we associate
$\{\vec{v} \in \Reg(T) \mid v_i \geq \mu_i, i = (j \mod k) + 1\}$; to its left child $T_\prec$, the region $\{\vec{v} \in \Reg(T) \mid v_i < \mu_i, i = (j \mod k) + 1\} \cup M$, where $M$ is the set of all vectors
$\vec{v} \in \Reg(T)$ with $v_i = \mu_i$ if $B_{\prec}$ (as defined in \Cref{alg:build_tree}) contains any such vector and $\emptyset$ otherwise. We also write $\Reg(\vec{u})$ for the 
region $\{\vec{v} \in \mathbb{N}^k \mid \vec{u} \leq \vec{v}\}$, \ie the \emph{upward closure} of
$\vec{u}$.   The following is our search algorithm:

\begin{algorithm}
\caption{SearchKDTree($T,\vec{u}$)}\label{alg:search_tree}
\begin{algorithmic}[1]
\If{$T$ is a leaf storing $\vec{v}$}
    \State Return whether $\vec{u} \leq \vec{v}$
\EndIf
    \If{$\Reg(T_\succ) \subseteq \Reg(\vec{u})$} \label{line:tsucc-inc}
        \State Return true
    \EndIf
        \State $R_\succ \gets$ SearchKDTree($T_\succ,\vec{u}$), $R_\prec \gets $ false \label{line:tsucc}
    \If{$\Reg(T_\prec) \cap \Reg(\vec{u}) \neq \emptyset$}
        \State $R_\prec \gets$ SearchKDTree($T_\prec,\vec{u}$) \label{line:tprec}
    \EndIf
    \State Return $R_\succ \lor R_\prec$
\end{algorithmic}
\end{algorithm}

By construction of a \kdtree 
$T$, if $\Reg(T) \subseteq \Reg(\vec{u})$, 
all vectors $\vec{v}$ stored at the leaves of $T$ are such that $\vec{u} \leq
\vec{v}$. There are no empty \kdtrees, so this instance of the membership
problem is positive. Conversely, if $\Reg(T) \cap \Reg(\vec{u}) = \emptyset$,
for all vectors $\vec{v}$ stored at its leaves we have $\vec{u} \not\leq
\vec{v}$. 
Since $\Reg(\vec{u})$ is upward closed:
\begin{itemize}
    \item $\Reg(T_\succ) \cap \Reg(\vec{u}) \neq \emptyset$ if $\Reg(T) \cap \Reg(\vec{u}) \neq \emptyset$, and
    \item $\Reg(T_\prec) \subseteq \Reg(\vec{u})$ only if $\Reg(T) \subseteq \Reg(\vec{u})$. 
\end{itemize}
It follows that \Cref{alg:search_tree} returns true if and only if the
instance is positive. It remains to study its complexity.\footnote{It is
interesting to compare the $k^2$ factor in the bound with the linear one $k$
informally claimed by Chan in the introduction
of~\cite{DBLP:journals/dcg/Chan19}. We did not find a reference for the
claimed $O(k m^{1-\nicefrac{1}{k}})$ bound, nor were we able to re-prove it
ourselves.} 

\begin{theorem}\label{thm:kdt-membership}
    There is a \kdtree algorithm that solves
    membership 
    in time 
    $O(\min(km,k^2m^{1-\nicefrac{1}{k}}))$.
\end{theorem}
\begin{proof}
    The first term in the minimum comes from the fact that the tree has at most
    $2m$ nodes. The second term will follow from bounding
    the number of nodes the algorithm treats by
    $O(km^{1-\nicefrac{1}{k}})$. For the remaining $k$ factor in both terms, we observe that, on leaves, one does need to compare $k$ numbers against
    each other. However, on internal nodes of the tree, a recursive call can be
    made to do only a constant amount of work. It is easy to see that checking
    whether the intersection of the regions is empty can be done in constant time:
    we only need to see how $\mu_i$ and $u_i$ compare, where $\vec{\mu}$ is the
    median used to split at this node of the tree and $i = (j \mod k) + 1$ with
    $j$ the depth. Slightly less obvious is the fact that one can determine the
    region inclusions from line~\ref{line:tsucc-inc} in constant time. This can
    be done by keeping track of $k$ variables which store the lower bounds
    of the region of the current node and a counter $c$ initially set to the number of strictly positive components of $\vec{u}$. On
    each recursive call, only one variable needs to be updated and if it becomes
    larger than the corresponding component of $\vec{u}$ we decrement $c$, \ie $c
    \gets c - 1$. Clearly, the inclusion holds if and only if $c = 0$.

    It remains to argue that the number of nodes treated is indeed $O(km^{1-\nicefrac{1}{k}})$.
    Say a region $\Reg(T)$ is
    $i$-interesting, for $1 \leq i \leq k$, if there are $\vec{v},\vec{w} \in
    \Reg(T)$ such that $v_i < u_i \leq w_i$. Essentially, not all vectors
    stored in $T$ are guaranteed to be larger than $\vec{u}$, with respect to
    the $i$-th coordinate, but neither are they all strictly smaller. The
    algorithm makes a recursive call on $T$ only if $\Reg(T)$ is
    $i$-interesting for some $i$. Hence, if for all $i$ we can bound the
    number of nodes of the tree with an $i$-interesting region by
    $O(m^{1-\nicefrac{1}{k}})$, the bound will follow. From here on, our
    argument follows~\cite[Ch. 4.10]{brass_2008}.

    Let $1 \leq i \leq k$ be arbitrary and consider a node $T$ at depth $(i-1)$ from the root. Since $B_\prec$ and $B_\succ$ are constructed based on the $i$-th coordinate, we have that at most one subtree among $\Reg(T_\prec)$ and $\Reg(T_\succ)$ is $i$-interesting. This dichotomy holds again for the subtrees rooted $k$ levels down since, once more, the vectors are split based on the $i$-th coordinate. 
    Now let $a_j$ be the number of nodes at level $0 \leq j \leq \log m$ from the root with an $i$-interesting region. From the analysis above we have that $a_j \leq 2^{\lfloor (1 - \nicefrac{1}{k})j\rfloor}$ since we double the number of nodes at every level except when $i = (j \mod k) + 1$.
    For the total number of nodes (across all levels) with $i$-interesting regions we get the following.
    \[
        a_0 + a_1 + \dots + a_{\log m} \leq 2 (2^{(1-\nicefrac{1}{k})(\log m)}) = 2 m^{1-\nicefrac{1}{k}}
    \]
    Intuitively, every $k$ layers, the doubling does not happen, so we lose a factor of $m^{\frac{1}{k}}$ nodes out of the total $2m$ nodes.
    Since the right-hand side of the equation above is $O(m^{1-\nicefrac{1}{k}})$, as required, this concludes the proof.\qed
\end{proof}

Note that the bound from \cref{thm:kdt-membership} simplifies to
$O(k^2m^{1-\nicefrac{1}{k}})$ if $m \geq 2^{k \log k}$. Henceforth, to
simplify our analysis, we will assume this inequality holds. Nevertheless, for
the claims, we state the bounds in their full generality.

\subsection{The union operation}\label{sec:ac-union}
Let $B \in
\mathbb{N}^k$ be a second antichain with $n$ vectors representing a $k$-dimensional downset.  Our
\kdtree-based algorithm for the union operation follows the one
proposed in \Cref{sec:list-union}. The main difference is that we leverage the
complexity of the \kdtree-based membership problem to obtain a different
complexity bound. 

By~\Cref{lem:union}, to compute $C = \maxac{\close{A} \cup
\close{B}}$ it suffices to remove from $A$ those elements dominated by some element in $B$ and to union them with the elements of $B$ that are not (strictly) dominated by some element in $A$.
The strict domination check can be realized with a
modification of our membership algorithm. The checks take time
$O(k^2mn^{1-\nicefrac{1}{k}})$ and $O(k^2m^{1-\nicefrac{1}{k}}n)$,
respectively. Finally, constructing the \kdtree for $C$ takes time
$O((m+n)\log(m+n))$.

\begin{theorem}
    There is a \kdtree-based algorithm for the union operation that runs in
    time $O(kn\min(m,km^{1-\nicefrac{1}{k}}) +n\log{n})$ assuming (w.l.o.g.) that $m \leq n$. 
\end{theorem}
Observe that, even if $2^{k \log k} \leq m, n$, this bound is \emph{not
always} better than the one we get for our list-based algorithm. Indeed, when
$\log(m+n)$ is larger than $km$ or $kn$, the last summand is
already worse than $kmn$. If $m \leq n$, then this cannot happen when, for
instance, $n \leq 2^m$.

\subsection{The intersection operation}
For intersection, we follow~\Cref{sec:list-inter} except that we use our \kdtree-based membership algorithm. 
To compute $C = \maxac{A \sqcap B}$, we first construct a \kdtree for the collection of meets obtained from $A$ and $B$. Note that this may not be a set. Nevertheless, the collection has size at most $mn$. Next, we use the tree to check, for all vectors it encodes, whether they are strictly smaller than some other vector in the tree. If the answer is negative, we add them to a new collection of nondominated meets. Now, to remove duplicates, we sort the collection lexicographically, this can be done in time $O(kmn \log (mn))$, and traverse it in search for consecutive copies of the same vector, this can be done in time $O(kmn)$. Finally, we construct a second \kdtree for the set of vectors. 

Building the first \kdtree can be done in time $O(kmn + mn \log(mn))$. For each vector in the set, the (strict) membership checks can be done in time $O(k^2(mn)^{1-\nicefrac{1}{k}})$, yielding a total of $O(k^2(mn)^{2-\nicefrac{1}{k}})$ checks. After removing duplicates in time $O(kmn \log (mn))$, we construct the second \kdtree in time $O(mn \log (mn))$. The total complexity is summarized in the result below.
\begin{theorem}
    There is a \kdtree-based algorithm for the intersection operation that
    runs in time $O(kmn \min(mn,
    k(mn)^{1-\nicefrac{1}{k}}))$.
\end{theorem}
This is better than 
our list-based algorithm, assuming $2^{k \log k} \leq m, n$.

A complementary behavior of $k$-d trees and sharing trees emerges from the fact that \kdtrees perform better when a number does not occur in several vectors in the same dimension, whereas in sharing trees, repetition of a number in a dimension enables ``sharing'', and hence decreases the size of the tree.

\subsection{Discussion: Theory}\label{sec:discussion}
The following table summarizes the complexity bounds of the list and
\kdtree-based algorithms from the previous sections. To recall, $m$ and $n$ are the sizes of antichains $A$ and $B$, respectively, of dimension $k$, and $W$ is the largest integer occurring among all vectors in $A$ and $B$. To simplify, we assume $m\leq n$.
\begin{center}
  \begin{tabular}{p{1.89cm}|p{1.55cm}|p{4.95cm}|p{3.3cm}}
    \bf \centering Operation & \bf\centering Lists & \bf\centering \kdtrees & \multicolumn{1}{c}{\bf\centering Sharing trees} \\
    \hline
    Membership & $O(km)$ & $O(\min(km,k^2m^{1-\nicefrac{1}{k}}))$ & $O(km)$ \\
    Union & $O(kmn)$ & $O(kn\min(m,km^{1-\nicefrac{1}{k}}) +n\log{n})$ & $O(kn(m+\min(n, W))$ \\
    Intersection & $O(km^2n^2)$ & $O(kmn \min(mn,
    k(mn)^{1-\nicefrac{1}{k}}))$ & $O(km^2n^2)$
  \end{tabular}
\end{center}

\noindent Our analysis confirms the
empirical findings of~\cite{DBLP:conf/tacas/CadilhacP23}: \kdtrees are not
always a better data structure than lists when manipulating antichains. However,
based on the remarks at the end of~\Cref{sec:ac-union}, one could conclude that
dynamically switching from lists to \kdtrees when 
\( 2^{k \log k} \leq m \leq n \leq 2^m \) can result in a good tradeoff. We implemented this, but unfortunately the size-to-dimension ratio of most antichains in our experiments does not trigger a switch to \kdtrees.



\section{Experiments}\label{sec:experiments}
We implemented several variations of the list-, \kdtree-, and sharing-tree-based
algorithms using a generic library for partially-ordered sets of vectors.%
\footnote{Links to our source code and benchmarks will be made
  available after acceptance.} 
Tests ``in a vacuum,'' where the data structure operations are benchmarked on
random data, exhibit the expected behavior dictated by their theoretical
complexity.  We do not report on these unsurprising conclusions (appearing in
the appendix).
Our
main interest lies in tests ``in the field,'' that is, in applications that rely
on antichains; one specific aim is to establish whether the conditions for
\kdtrees or sharing trees to outperform list-based downsets are met in practice.
We focus on two such applications: LTL-realizability and parity game solving,
relying on benchmark sets used in authoritative competitions.  For both
applications, we formally present the computational task at hand, the
downset-based algorithm to solve it, and experimental results.  We also study
the ratio of size vs.\ dimension within the benchmark sets used in these
applications.

All the following experiments were carried on an Intel\textsuperscript{\textregistered}
Core\textsuperscript{\texttrademark} i7-8700 CPU @ 3.20GHz paired with 16GiB of memory. 

\subsection{LTL-realizability}
\label{sec:ab}

\subsubsection{The task.}

Let $I$ and $O$ be disjoint and finite sets of \emph{input} and \emph{output}
\emph{propositions}.  A \emph{linear-temporal-logic} (LTL) formula over
$P = I \cup O$ specifies temporal dependencies between truth values of the
propositions.  Formulas in LTL are constructed from the propositions, the usual
Boolean connectives, and temporal operators ``next'', ``eventually'',
``always'', and ``until'', with their intuitive semantics. (We refer the reader
to \cite{DBLP:books/daglib/0020348} for the formal syntax and semantics of LTL.)
It is well known that the set $\mathrm{Words}(\varphi)$ of all words, over valuations
$2^{P}$ of the propositions, that satisfy a given LTL formula $\varphi$ can be
``compiled'' into an infinite-word automaton. In particular, one can construct a
non-deterministic automaton $\mathcal{N}$ with a \emph{B\"uchi acceptance condition} such
that its language is exactly $\mathrm{Words}(\varphi)$. The B\"uchi acceptance
condition stipulates that infinite runs of the automaton are accepting if they
visit accepting states infinitely often.

LTL realizability can be defined in terms of the aforementioned
automaton. Namely, given the nondeterministic B\"uchi automaton $\mathcal{N}$ constructed
from an input LTL formula $\varphi$ over $P$, we have an \emph{input player} and an
\emph{output player} take turns choosing truth values for $I$ and $O$. This
induces an infinite word over $2^P$. The winner of the game depends on whether
the word is in the language of $\mathcal{N}$: the input player wins if it is \emph{not} in
the language, otherwise the output player wins. (There are other ways of
defining realizability using membership in the language and other acceptance
conditions, but this particular one makes it easier to present the algorithm
below.) The computational task lies in determining if the input player has a
winning strategy for this game.

\subsubsection{The algorithm.}

We succinctly present the downset-based approach of Filiot et
al.~\cite{DBLP:journals/fmsd/FiliotJR11} to solving the task at hand.  Fix a
\emph{Büchi automaton} $\mathcal{N}=(Q, q_0, \delta, B)$ with \(Q\) a set of states,
\(q_0\) the initial state, \(\delta\) the transition relation that uses valuations
$2^P$ as labels, and \(B \subseteq Q\) the set of B\"uchi states.  We will be interested
in \emph{vectors over $Q$}, i.e. elements in \(\mathbb{Z}^Q\) mapping states to integers,
to encode the number of visits to Büchi states---recall the input player wants to
avoid there being a run that visits these infinitely often. We will write
\(\vec{v}\) for such vectors, and \(v_q\) for its value for state \(q\).  In practice,
these vectors will range into a finite subset of \(\mathbb{Z}\), with \(-1\) as an implicit
minimum value (\ie \((-1) - 1\) is still \(-1\)) and an upper bound $k$
that can be thought of as a hyperparameter of the algorithm.

For a vector \(\vec{v}\) over \(Q\) and a valuation \(x\), we define a function that takes
one step back in the automaton, decreasing components that have seen B\"uchi
states. Write \(\chi_B(q)\) for the function mapping a state \(q\) to \(1\) if
\(q \in B\), and \(0\) otherwise.  We then define \(\bwd(\vec{v}, x)\) as the vector
over $Q$ that maps each state \(p \in Q\) to:
\[\min_{(p, x, q) \in \delta} \left(v_q -
  \chi_B(q)\right)\enspace,
\]
and we generalize this to sets: \(\bwd(S, x) = \{\bwd(\vec{v}, x) \mid \vec{v}
\in S\}\).
For a set \(S\) of vectors over $Q$ and a valuation $i \in 2^I$ of the inputs, define:
\[\cpre_i(S) = S \cap \bigcup_{o \in 2^O} \bwd(S, i \cup o)\enspace.\]

It is proved in~\cite{DBLP:journals/fmsd/FiliotJR11} that iterating $\cpre$  converges to
a fixed point that is independent from the order in which the valuation of the inputs is 
selected.  We define \(\cpre^*(S)\) to be that set.

All the sets that we manipulate above are \emph{downsets}. Now, for any $k > 0$,
if there is a \(\vec{v} \in \cpre^*(\{i \in \mathbb{N} \mid i \leq k\}^Q)\) with
\(v_{q_0} \geq 0\) then the input player has a winning strategy. Conversely, there is
a large enough value of $k$ such that if the condition does not hold then the
output player has a winning strategy.

\subsubsection{Experimental results.}

The above algorithm was implemented as the tool
Acacia-Bonsai~\cite{DBLP:conf/tacas/CadilhacP23}, relying on our generic library
for partially-ordered sets.  We considered the benchmarks used in the yearly
competition in LTL-realizability, SYNTCOMP~\cite{syntcomp}.  These consist in
1048 LTL formulas, of which the best LTL tools solve 90\% in under a
second.  We present the experimental results as a survival plot, indicating how
many tests are solved (x-axis) within a time limit (y-axis, time per test).  In
particular, the lower the curve, the better.  To reduce clutter, we focus on the
benchmarks that took the longest: 500 benchmarks are not shown, with all
implementations solving each of them in less than 0.1 seconds.  The tests were
executed with a 60-second timeout.  We observe, crucially, that the dynamic
switching between the two data structures follows closely the list-based
implementation, witnessing the fact that the threshold provided by the theory at
which \kdtrees are more advantageous (see \Cref{sec:discussion}) is rarely crossed.  To illustrate this, we
studied the \emph{ratio} of set size vs.\ dimension; it is also displayed here as
a survival plot, indicating how many sets (x-axis) have a ratio below a certain
value (y-axis).  The sets we considered are all the different values of \(\cpre\),
which, over all the benchmarks, amounts to 67,143 sets.

\begin{figure}[h!]
  \centering
  \begin{subfigure}[t]{0.5\textwidth}
    \includegraphics[height=5cm]{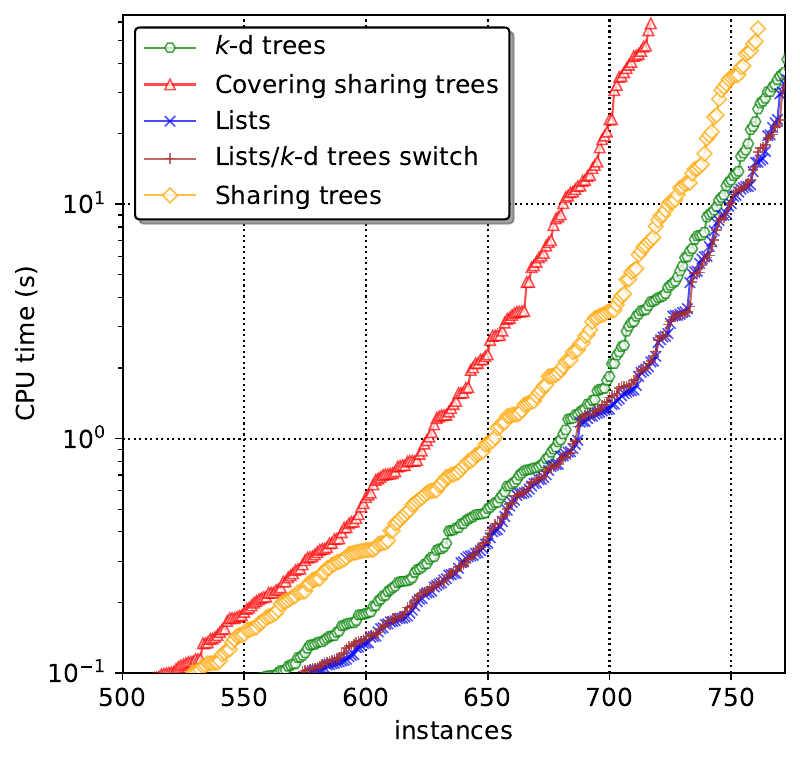}
    \caption{Benchmarking LTL-realizability.}
    \label{fig:ab-cactus}
  \end{subfigure}%
  \begin{subfigure}[t]{0.5\textwidth}
    \includegraphics[height=5cm]{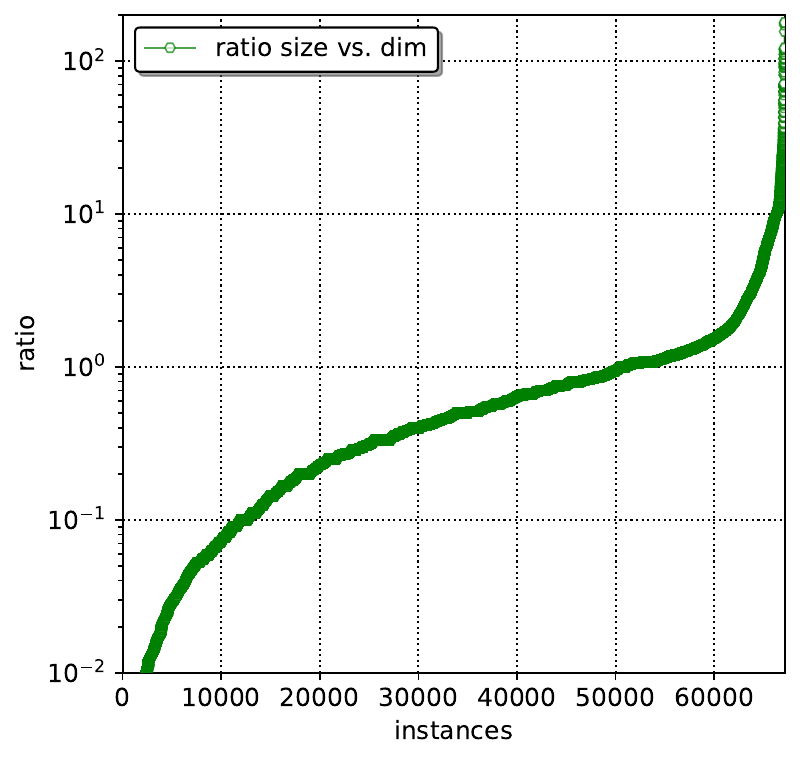}
    \caption{Ratio of size vs.\ dimension.}
    \label{fig:rat}
  \end{subfigure}
  \caption{Survival plots of downset-based LTL-realizability.}
  \label{fig:ab}
\end{figure}

The \kdtree implementation solved 779 test cases, while the list implementations
solved a strict superset of 787 cases.  Sharing tree
implementations performed uniformly worse than all the others, solving 717 test
cases for covering sharing trees, and 761 for sharing trees.  These are in turn
subsets of the ones solved by \kdtree, except for one case.


These results indicate that the size of the antichains and their dimensions are
too low to translate into an edge for the \kdtree structure, but that they
nonetheless achieve similar performances.  This is clearly indicated by
\Cref{fig:rat}: roughly 95\% of sets that are created in the algorithm are of
size that is bounded by twice the dimension.  Larger ratios are only found in
161 of the 787 solved cases.  We present, in Appendix~\ref{sec:5percent},
the same graphics as above but focusing on the instances these 161 instances
that \emph{do} induce set sizes that are at least double the dimension.  No
notable difference in performance is to be reported.

When introduced in 2023, Acacia-Bonsai also included an implementation of
\kdtrees.  We present, in Appendix~\ref{sec:old-kdtrees}, how that
implementation fairs against our current implementation: it would be performing
slightly worse than sharing trees. Concretely, on the set of benchmarks presented in the next subsection, optimized algorithms based on the theoretical insights presented in this paper have led to half an order of magnitude improvement in the performance of the data structure. We speculate that similar improvements are still possible for our implementation of (covering) sharing trees.

\subsection{Parity game solving}

\subsubsection{The task.} A \emph{parity game} $\mathcal{G}$ is a tuple $(V_0,V_1,E,p)$ where $V_0$ and
$V_1$ are disjoint sets of vertices with $V = V_0 \cup V_1$, $E \subseteq V
\times V$ is a set of directed edges, and $p \colon V \to \mathbb{N}$ assigns a
\emph{priority} to each vertex.

In such games, we usually assume an \emph{even player} controls $V_0$ while an
\emph{odd player} controls $V_1$. The players select outgoing edges from their
vertices, and it induces an infinite path from a given starting vertex. The
winner of the game depends on the maximal vertex priority appearing infinitely
often along the path: the even player wins if it is even, otherwise the odd
player wins.  The computational task of parity game solving is to determine for
each vertex \(v\) if the controller of the vertex has a winning strategy in the
game starting in \(v\).

\subsubsection{The algorithm.} We present a new algorithm---to the best
of our knowledge---based on a construction of Bernet et
al.~\cite{DBLP:journals/ita/BernetJW02}, to determine the winner of a parity
game via manipulation of downsets. Our presentation follows the vocabulary
of~\cite[Ch. 6]{DBLP:books/cu/11/0001R11}.  Let us fix a parity game
$(V_0,V_1,E,p)$ and an initial vertex $v_0$. Also, write
$d = \lceil \max_{v \in V} p(v) / 2 \rceil$ and
$p^{-1} \colon \{0, 1, \dots, 2d\} \to 2^V$ for the mapping from every priority to the
set of all vertices labelled by it.
We use vectors $\vec{c} \in \mathbb{N}^d$ to keep track of the number of
visits to odd-priority vertices. More precisely, we consider $\vec{c} \in
\mathbb{Z}^d$ such that $-1 \leq c_i \leq n_i$ for all $1 \leq i \leq d$,
where $n_i = |p^{-1}(2i-1)|$. We adopt the convention that $c_i - a = -1$ for
all $a \geq 1 + c_i$ or if $c_i = -1$, and $c_i + b = n_i$ for all $b \geq n_i
- c_i$. The intuition is that $c_i$ keeps track of the number of visits to
vertices with odd priority $2i - 1$ and having a value of $-1$ means that
reaching a value of $n_i$, hence a simple cycle with that odd
priority, is unavoidable.

For a vector $\vec{c}$ and a vertex $v \in V$, we define an operator to
obtain a predecessor vector (recall that we are counting visits to
odd-priority vertices). Formally,
\[
    \bwd(\vec{c},v) = \begin{cases}
        \vec{c} - \vec{e_i} & \text{if } p(v) = 2i - 1\\
        \vec{c} + \sum_{j=1}^i n_{j} \vec{e_{j}} & \text{if } p(v) = 2i
    \end{cases}
\]
where $\vec{e_i}$ is a vector with $1$ as its $i$-th component and zeros
elsewhere. We lift the operator to sets $C$ of vectors as follows. 
\[
\bwd(C,v) = \{ \bwd(\vec{c},v) \mid \vec{c} \in C\}
\]

Consider a mapping $\mu \colon V \to 2^{\mathbb{N}^d}$ from vertices to sets of
vectors. We 
introduce an update operation on such mappings to
compute over-approximations of the states from which the even player wins. The
initial mapping $\mu_0$ is the closure of the vector of $n_i$, i.e. it assigns the following set to all vertices $v \in V$:
\[
   \mu_0(v) = \close{\{(n_1, n_3, \dots, n_{2i+1}, \dots)\}}\enspace.
\]
Intuitively, this is the worst possible situation for the even player without
having reached the maximal number of visits to some odd-priority vertex. The
update operation $\cpre$, for a given mapping $\mu$, outputs a new mapping
$\nu$ such that:
\begin{gather*}
\nu(u) = \mu(u) \cap C\enspace, \text{ where } \quad
\begin{aligned}
  C & = \begin{cases}
    \bigcup_{(u,v) \in E} \close{\bwd(\mu(v), u)}  & \text{if } u \in V_0\\
    \bigcap_{(u,v) \in E} \close{\bwd(\mu(v), u)}  & \text{if } u \in V_1.
  \end{cases}
\end{aligned}
\end{gather*}

It can be shown that the mappings obtained by iterating $\cpre$ starting from
$\mu_0$ converge. Write $\cpre^*$ for that mapping. We have~\cite[Lemma 6.4]{DBLP:books/cu/11/0001R11} that
for all $v \in V$, there is some $\vec{c} \in \cpre^*(v)$ such that $0 \leq c_i 
\leq
n_i$ for all $1 \leq i \leq d$ if and only if even player wins the parity game when
starting from $v$.

\subsubsection{Experimental results.}

We again rely on the SYNTCOMP24 benchmarks, which has a competition track for
parity game solvers, and augment these benchmarks with the ones provided by
Keiren~\cite{keiren15}.  We implemented our downset-based algorithm in the tool
Oink~\cite{vanDijk18a}, a tool developed to provide a uniform framework for the
comparison of parity game solvers.  The implementation is agnostic to the
downset implementation, allowing for easy comparison of the underlying data
structure.  The 779 benchmarks were each executed with a 60-second timeout and
10GB memory limitation.  The survival plot appears below.

We studied the ratio of set size vs.\ dimension, listing the size of the set
\(C\) (as used in the definition of \(\nu\) above) every time it is computed; this is
a grand total of 1,106,234,004 sets.  For 98\% of them, the ratio was smaller
than 0.03; in fact, out of the billion sets processed by our algorithm, only
746,139 are of size greater or equal to 3.  These are concentrated on only 22
out of 1024 benchmarks and we also display the survival plot corresponding to
these tests only, since they are more likely to favor implementations that
perform well with large-sized sets.

\begin{figure}[h!]
  \centering
  \begin{subfigure}[t]{0.5\textwidth}
    \includegraphics[height=5cm]{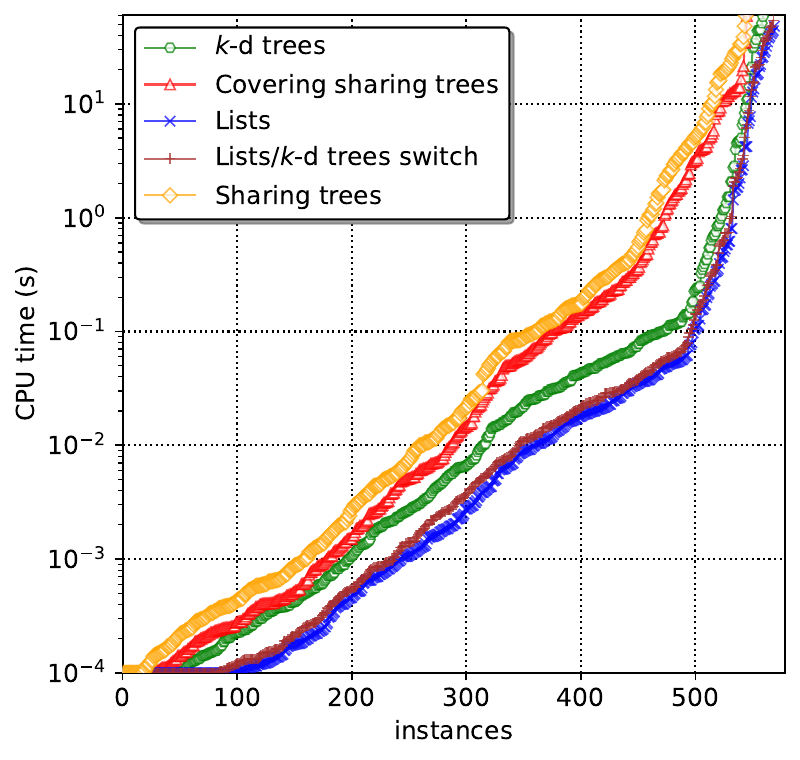}
    \caption{Over all 779 games.}
  \end{subfigure}%
  \begin{subfigure}[t]{0.5\textwidth}
    \includegraphics[height=5cm]{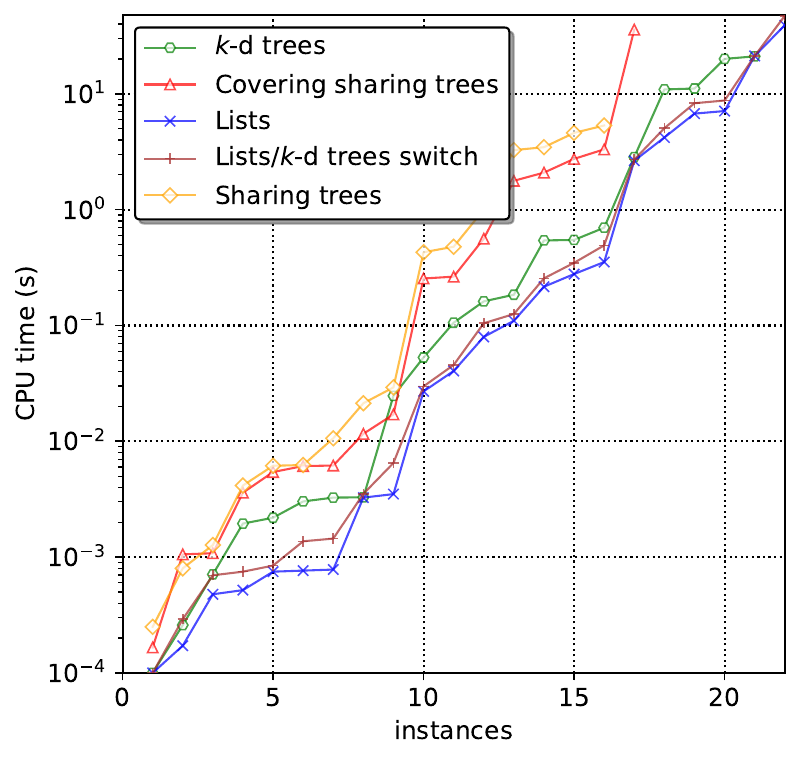}
    \caption{Over games with nontrivial downsets.}
  \end{subfigure}
  \caption{Survival plots of downset-based parity-game solving.}
  \label{fig:oink}
\end{figure}

On these smaller downsets, we see that covering sharing trees does better than
sharing trees. It would seem that when the size of the antichain of maximal
elements of an antichain is not much smaller than the downset itself, the
approximate domination check of covering sharing trees shines
through. Nevertheless, it is clear that list-based algorithms are also best for
the parity-game use case: over all 779 games, they solve 569 test cases, while
\kdtrees solve 559, and (covering) sharing trees 544.

In Appendix \ref{sec:mem}, we also plot the memory consumption of the data
structures, which clearly indicates (covering) sharing trees are the least space
efficient. This supports our earlier claim that (covering) sharing trees have a
worse performance due to them storing larger sets than the other data
structures.

\section{Conclusion}

We provided a theoretical analysis of two data structures for downsets of
natural vectors: list-, sharing-tree-, and \kdtree-based.  We identified when \kdtrees should
outperform the others and provided experiments showing that uses ``in the
field'' are not conducive to cases where \kdtrees outperform the humble
list-based implementation. We posit other antichain-based tools like the Petri-net safety-checking tool MIST \cite{gantyAlgorithmesStructuresDonnees2002} may benefit from using list-based antichains.

For future work, 
it would be interesting to provide average-case complexity bounds for the operations we have studied. Most naive approaches for an average-case analysis of antichain-manipulations seem to require tighter bounds on \emph{Dedekind numbers} (intuitively, the number of different antichains for a given norm bound and dimension) than the ones we found in the literature (see, e.g.~\cite{falgasravry2023dedekinds}). Additionally, our study of antichain size vs. dimension may indicate sharing-tree based antichains benefit from radix-tree like compression~\cite{DBLP:conf/wia/HolubC02}.


\bibliographystyle{splncs04}
\bibliography{bibliography}

\begin{thebibliography}{10}
\providecommand{\url}[1]{\texttt{#1}}
\providecommand{\urlprefix}{URL }
\providecommand{\doi}[1]{https://doi.org/#1}

\bibitem{DBLP:books/daglib/0020348}
Baier, C., Katoen, J.: Principles of model checking. {MIT} Press (2008)

\bibitem{kdtree_book}
de~Berg, M., Cheong, O., van Kreveld, M.J., Overmars, M.H.: Computational
  geometry: algorithms and applications, 3rd Edition. Springer (2008),
  \url{https://www.worldcat.org/oclc/227584184}

\bibitem{DBLP:journals/ita/BernetJW02}
Bernet, J., Janin, D., Walukiewicz, I.: Permissive strategies: from parity
  games to safety games. {RAIRO} Theor. Informatics Appl.  \textbf{36}(3),
  261--275 (2002). \doi{10.1051/ITA:2002013},
  \url{https://doi.org/10.1051/ita:2002013}

\bibitem{DBLP:journals/iandc/BerwangerCWDH10}
Berwanger, D., Chatterjee, K., Wulf, M.D., Doyen, L., Henzinger, T.A.: Strategy
  construction for parity games with imperfect information. Inf. Comput.
  \textbf{208}(10),  1206--1220 (2010)

\bibitem{Blumer}
Blumer, A., Blumer, J., Ehrenfeucht, A., Haussler, D., McConnell, R.M.: Linear
  size finite automata for the set of all subwords of a word - an outline of
  results. Bull. {EATCS}  \textbf{21},  12--20 (1983)

\bibitem{brass_2008}
Brass, P.: Advanced Data Structures. Cambridge University Press (2008).
  \doi{10.1017/CBO9780511800191}

\bibitem{DBLP:conf/tacas/CadilhacP23}
Cadilhac, M., P{\'{e}}rez, G.A.: Acacia-bonsai: {A} modern implementation of
  downset-based {LTL} realizability. In: {TACAS} {(2)}. Lecture Notes in
  Computer Science, vol. 13994, pp. 192--207. Springer (2023)

\bibitem{DBLP:journals/dcg/Chan19}
Chan, T.M.: Orthogonal range searching in moderate dimensions: k-d trees and
  range trees strike back. Discret. Comput. Geom.  \textbf{61}(4),  899--922
  (2019). \doi{10.1007/S00454-019-00062-5},
  \url{https://doi.org/10.1007/s00454-019-00062-5}

\bibitem{DBLP:reference/mc/2018}
Clarke, E.M., Henzinger, T.A., Veith, H., Bloem, R. (eds.): Handbook of Model
  Checking. Springer (2018). \doi{10.1007/978-3-319-10575-8},
  \url{https://doi.org/10.1007/978-3-319-10575-8}

\bibitem{Daciuk}
Daciuk, J., Mihov, S., Watson, B.W., Watson, R.E.: Incremental construction of
  minimal acyclic finite state automata. Comput. Linguistics  \textbf{26}(1),
  3--16 (2000). \doi{10.1162/089120100561601},
  \url{https://doi.org/10.1162/089120100561601}

\bibitem{DBLP:journals/siamcomp/DaskalakisKMRV11}
Daskalakis, C., Karp, R.M., Mossel, E., Riesenfeld, S.J., Verbin, E.: Sorting
  and selection in posets. {SIAM} J. Comput.  \textbf{40}(3),  597--622 (2011).
  \doi{10.1137/070697720}, \url{https://doi.org/10.1137/070697720}

\bibitem{DBLP:journals/sttt/DelzannoRB04}
Delzanno, G., Raskin, J., Begin, L.V.: Covering sharing trees: a compact data
  structure for parameterized verification. Int. J. Softw. Tools Technol.
  Transf.  \textbf{5}(2-3),  268--297 (2004). \doi{10.1007/S10009-003-0110-0},
  \url{https://doi.org/10.1007/s10009-003-0110-0}

\bibitem{vanDijk18a}
van Dijk, T.: Oink: An implementation and evaluation of modern parity game
  solvers. In: TACAS. Lecture Notes in Computer Science, vol. 10805, pp.
  291--308. Springer (2018). \doi{10.1007/978-3-319-89960-2\_16},
  \url{https://doi.org/10.1007/978-3-319-89960-2\_16}

\bibitem{knor}
van Dijk, T., van Abbema, F., Tomov, N.: Knor: reactive synthesis using oink.
  In: Finkbeiner, B., Kov{\'{a}}cs, L. (eds.) Tools and Algorithms for the
  Construction and Analysis of Systems - 30th International Conference, {TACAS}
  2024, Held as Part of the European Joint Conferences on Theory and Practice
  of Software, {ETAPS} 2024, Luxembourg City, Luxembourg, April 6-11, 2024,
  Proceedings, Part {I}. Lecture Notes in Computer Science, vol. 14570, pp.
  103--122. Springer (2024). \doi{10.1007/978-3-031-57246-3\_7},
  \url{https://doi.org/10.1007/978-3-031-57246-3\_7}

\bibitem{DBLP:conf/tacas/DoveriGH23}
Doveri, K., Ganty, P., Hadzi{-}Dokic, L.: Antichains algorithms for the
  inclusion problem between {o}mega-{VPL}. In: {TACAS} {(1)}. Lecture Notes in
  Computer Science, vol. 13993, pp. 290--307. Springer (2023)

\bibitem{DBLP:journals/corr/abs-0902-3958}
Doyen, L., Raskin, J.: Antichains for the automata-based approach to
  model-checking. Log. Methods Comput. Sci.  \textbf{5}(1) (2009)

\bibitem{DBLP:books/cu/11/0001R11}
Doyen, L., Raskin, J.: Games with imperfect information: theory and algorithms.
  In: Lectures in Game Theory for Computer Scientists, pp. 185--212. Cambridge
  University Press (2011)

\bibitem{falgasravry2023dedekinds}
Falgas-Ravry, V., Räty, E., Tomon, I.: Dedekind's problem in the hypergrid
  (2023), \url{https://arxiv.org/abs/2310.12946}

\bibitem{DBLP:journals/fmsd/FiliotJR11}
Filiot, E., Jin, N., Raskin, J.: Antichains and compositional algorithms for
  {LTL} synthesis. Formal Methods Syst. Des.  \textbf{39}(3),  261--296 (2011)

\bibitem{gantyAlgorithmesStructuresDonnees2002}
Ganty, P.: Algorithmes et Structures de Donn{\'e}es Efficaces Pour La
  Manipulation de Contraintes Sur Les Intervalles (in {{French}}). Master's
  thesis, Universit{\'e} Libre de Bruxelles, Belgium (2002)

\bibitem{gantySymbolicDataStructure2007}
Ganty, P., Meuter, C., Delzanno, G., Kalyon, G., Raskin, J., {Van Begin}, L.:
  Symbolic {{Data Structure}} for sets of k-uples. Tech. Rep.~570,
  {Universit{\'e} Libre de Bruxelles, Belgium} (2007)

\bibitem{DBLP:journals/fmsd/HolikIRV20}
Hol{\'{\i}}k, L., Iosif, R., Rogalewicz, A., Vojnar, T.: Abstraction refinement
  and antichains for trace inclusion of infinite state systems. Formal Methods
  Syst. Des.  \textbf{55}(3),  137--170 (2020)

\bibitem{DBLP:conf/wia/HolubC02}
Holub, J., Crochemore, M.: On the implementation of compact dawg's. In:
  Champarnaud, J., Maurel, D. (eds.) Implementation and Application of
  Automata, 7th International Conference, {CIAA} 2002, Tours, France, July 3-5,
  2002, Revised Papers. Lecture Notes in Computer Science, vol.~2608, pp.
  289--294. Springer (2002). \doi{10.1007/3-540-44977-9\_31},
  \url{https://doi.org/10.1007/3-540-44977-9\_31}

\bibitem{DBLP:journals/acta/HunterPR18}
Hunter, P., P{\'{e}}rez, G.A., Raskin, J.: Looking at mean payoff through foggy
  windows. Acta Informatica  \textbf{55}(8),  627--647 (2018)

\bibitem{syntcomp}
Jacobs, S., P{\'{e}}rez, G.A., Abraham, R., Bruy{\`{e}}re, V., Cadilhac, M.,
  Colange, M., Delfosse, C., van Dijk, T., Duret{-}Lutz, A., Faymonville, P.,
  Finkbeiner, B., Khalimov, A., Klein, F., Luttenberger, M., Meyer, K.J.,
  Michaud, T., Pommellet, A., Renkin, F., Schlehuber{-}Caissier, P., Sakr, M.,
  Sickert, S., Staquet, G., Tamines, C., Tentrup, L., Walker, A.: The reactive
  synthesis competition {(SYNTCOMP):} 2018-2021. Int. J. Softw. Tools Technol.
  Transf.  \textbf{26}(5),  551--567 (2024). \doi{10.1007/S10009-024-00754-1},
  \url{https://doi.org/10.1007/s10009-024-00754-1}

\bibitem{DBLP:conf/stacs/Jurdzinski00}
Jurdzinski, M.: Small progress measures for solving parity games. In: Reichel,
  H., Tison, S. (eds.) {STACS} 2000, 17th Annual Symposium on Theoretical
  Aspects of Computer Science, Lille, France, February 2000, Proceedings.
  Lecture Notes in Computer Science, vol.~1770, pp. 290--301. Springer (2000).
  \doi{10.1007/3-540-46541-3\_24},
  \url{https://doi.org/10.1007/3-540-46541-3\_24}

\bibitem{keiren15}
Keiren, J.J.A.: Benchmarks for parity games. In: Dastani, M., Sirjani, M.
  (eds.) Fundamentals of Software Engineering. pp. 127--142. Springer
  International Publishing, Cham (2015)

\bibitem{DBLP:journals/lmcs/LaveauxGW21}
Laveaux, M., Groote, J.F., Willemse, T.A.C.: Correct and efficient antichain
  algorithms for refinement checking. Log. Methods Comput. Sci.  \textbf{17}(1)
  (2021), \url{https://lmcs.episciences.org/7143}

\bibitem{DBLP:journals/tcs/Revuz92}
Revuz, D.: Minimisation of acyclic deterministic automata in linear time.
  Theor. Comput. Sci.  \textbf{92}(1),  181--189 (1992).
  \doi{10.1016/0304-3975(92)90142-3},
  \url{https://doi.org/10.1016/0304-3975(92)90142-3}

\bibitem{DBLP:conf/icfem/WangS0LDWL12}
Wang, T., Song, S., Sun, J., Liu, Y., Dong, J.S., Wang, X., Li, S.: More
  anti-chain based refinement checking. In: {ICFEM}. Lecture Notes in Computer
  Science, vol.~7635, pp. 364--380. Springer (2012)

\bibitem{zampunieris1997sharing}
Zampuni{\'e}ris, D.: The Sharing Tree Data Structure, Theory and Applications
  in Formal Verification. Ph.D. thesis, PhD thesis, Department of Computer
  Science, University of Namur, Belgium (1997)

\end{thebibliography}

\clearpage
\appendix

\section{Membership: best and worst-case scenarios}
Here we give intuition and examples of antichains showcasing the best- and worst-case scenarios for membership analysis in all three data structures. We have $k$ the dimension of the space, $m$ the number of vectors in the antichain $A$, also encoded in a \kdtree $T$ and approximated by a sharing tree $G$, and $\vec{u}$ the vector we want to check for membership in $\close{A}$. We consider two cases: one where the vector being queried is in the downset of the antichain, and one where it is not.
\begin{description}
    \item[The vector is in the downset.] For lists, in the best case, the vector could be smaller than the first vector in the antichain, and the algorithm terminates in $k$ steps, whereas the worst case is when the input vector is only smaller than the last vector in the list, in which case, the algorithm terminates after $km$ steps, i.e. after comparing against all the vectors in the antichain.

    For sharing trees, the best case is similar to that of lists: if the vector is dominated by the one encoded by the first branch of $G$, then the algorithm terminates in $k$ steps.
    In the worst case, the last branch checked is the one dominating $\vec{u}$. This means the whole sharing tree is traversed and by \Cref{lem:build-st} this takes $O(km)$.

    For $k$-d trees, the best case is when $u$ is the zero vector. In this case, $\Reg(T_\succ) \subseteq \Reg(\vec{u})$ and the algorithm terminates in $k$ steps. (The $k$ steps are actually for the initialization of the variable $c$, a single recursive call of the algorithm suffices afterwards! See the proof of \Cref{thm:kdt-membership}.) In the worst case, consider an antichain such that the last component of each vector is 0 except the one which is the rightmost leaf in the $k$-d tree representation. The vector for which we ask membership is $(0,0,\dots,0,1)$. On each recursive call, the algorithm will have to explore both branches of the tree, except at the level corresponding to the last dimension. This leads to the worst-case complexity of $O(k^2m^{1-\frac{1}{k}})$.

    \item[The vector is not in the downset.] For lists, the best and worst cases are the same: the algorithm takes $k m$ steps since the algorithm must compare every vector against $\vec{u}$.

    For sharing trees, the best case is when the first component of $\vec{u}$ is strictly larger than that of all vectors encoded by $G$, and since we keep successors ordered by their label, the data structure allows to exit and assert $\vec{u}$ is not in $\close{A}$, which is in $O(1)$. For the worst case, we again use \Cref{lem:build-st} since we have to traverse the whole sharing tree, and so we get that it takes $O(km)$ time.

    For $k$-d trees, the best case is when $u=(\ell,\ell,\dots,\ell)$ with $\ell$ strictly greater than the largest number occurring in any vector in the antichain. In this case, $\Reg(T_\prec) \cap \Reg(\vec{u}) = \emptyset$ will hold on all recursive calls, hence the algorithm will never enter a ``left'' branch so it will terminate in $O(k + \log m)$ steps.\footnote{Here, again, $k$ comes from the initialization of $c$.} The worst-case behavior can be obtained with the same example as before: suppose $A$ is such that the last component of all its vectors is $0$ and $\vec{u}=(0,0,\dots,0,1)$. Now the $k$-d tree algorithm will have to explore both branches of each subtree except at the level corresponding to the last dimension, which leads to the complexity of $O(k^2m^{1-\frac{1}{k}})$.
\end{description}




\section{Covering sharing tree algorithms}\label{sec:csts}
We now present the variant of sharing trees proposed in \cite{DBLP:journals/sttt/DelzannoRB04} to encode and manipulate antichains. The two main differences between these so-called \emph{covering sharing trees} and the sharing-tree based antichain algorithms we described in this work are the following.
\begin{enumerate}
    \item Covering sharing trees may encode more than just the antichain of maximal elements. Indeed, due to the \emph{approximate} domination checks proposed in \cite{DBLP:journals/sttt/DelzannoRB04}, some dominated vectors are not excluded from the encoded set. This makes comparing the complexity of algorithms on covering sharing trees with the others studied in this work rather hard. Indeed, the difference in size between the antichain of maximal elements of a downset and an arbitrary subset of its closure (which includes the maximal elements) is hard to bound in a nontrivial way.
    \item Algorithms to operate over two covering sharing trees, as proposed in \cite{DBLP:journals/sttt/DelzannoRB04}, are much more graph-based. Indeed, these resemble binary decision diagram algorithms. In contrast the algorithms proposed in this work revert to iterating over the encoded vectors for union and intersection and use the sharing tree representation only to speed up membership queries.
\end{enumerate}

Let $A \subset \mathbb{N}^k$ be the antichain of maximal elements for the downward-closed set of interest.
We use $\mathcal{L}(\mathcal{A})$ to denote the \textit{language} recognized by a DFA $\mathcal{A}$. Ideally, we want a sharing tree encoding $\mathcal{D}_{\lceil V\rceil}$ (i.e., the language accepting the minimal antichain containing $V$ in its downset), but computing this is known to be NP-hard in the size of the sharing tree~\cite{DBLP:journals/sttt/DelzannoRB04}.

Instead, we use a relaxation of $D_{\lceil V\rceil}$ called the \textit{simulation minimal} automaton $D_{\widetilde{V}}$ with the property that $\mathcal{L}({D_{\lceil V\rceil}})\subseteq\mathcal{L}({D_{\widetilde{V}}})\subseteq\mathcal{L}({D_V})=V$. The relaxation is based on a simulation-relation minimization described in~\cite{DBLP:journals/sttt/DelzannoRB04}.

\begin{definition}
    Let $n$ and $m$ be nodes of the $i$-th layer of trees $S$ and $T$, respectively. The node $n$ is \emph{(forward) simulated} by $m$, written \(n\xrightarrow{F} m\), if \(\mathrm{val}(n) \leq \mathrm{val}(m)\) and all successors of $n$ are simulated by some successor of $m$, that is, for all $s \in \mathrm{succ}(n)$ there exists $t \in \mathrm{succ}(m)$ such that \(s\xrightarrow{F} t\).
\end{definition}

We say a sharing tree is \textit{simulation minimal} if no child of a node simulates a sibling. This definition of simulation minimality gives a DFA that is potentially exponentially smaller than the starting set, even if the starting set is an antichain (see~\cite[Sec. 1.3.3]{zampunieris1997sharing} and~\cite[Prop. 2]{DBLP:journals/sttt/DelzannoRB04}).

\subsection{Growing sharing trees}
Now, \Cref{alg:build_stree} can be modified so that every addition of a new node is conditional on it not being (foward) simulated by a sibling that has already been added to the sharing tree. These additional checks mean that building the tree becomes more costly. Even more important is the fact that we no longer require the encoded set to be the antichain of maximal elements whose downward-closure is the original set. We allow vectors that are part of the downward-closure of another encoded vector as long as simulation-minimality holds.

\subsection{The membership problem}
This can still be realized using a DFS.

\subsection{The union operation}

Let $S$ be the sharing tree representing a downset $\close{A}$ and $T$ a sharing tree representing the downset $\close{B}$ of equal dimension $k \in \mathbb{N}_{>0}$. To compute the union of $S \cup T$, we adapt the union algorithm proposed by Zampuniéris \cite{zampunieris1997sharing}. On a node-level, it is shown that the union of two nodes $n$ and $m$ with the same value, or more precisely the subtrees rooted at those nodes, is a node with that shared value and its successors being the union of a successor of $n$ and one of $m$ if they have the same value and a simple copy of all remaining successors of both nodes. Note that due to definition \ref{def:st} (2), there is at most one successor of $m$ that has the same value as any of $n$.

Starting from the roots of $S$ and $T$, the successors of both nodes are iterated over. Due to the descending ordering, there are three main cases that can occur. In the first case, the value of the current successor in $S$ is larger than that in $T$, then the whole subtree rooted at this successor-node can be copied as successor of the resulting union-node. In the second case, all successors in one tree were visited, then the subtrees for the remaining successors in the other tree are copied. The last case is that both successors have the same value, then the union is recursively called for the two nodes and the resulting node added as successor of the union node. 

To ensure a simulation-minimal tree, all cases include a simulation-check before adding to the union-result. This means, copied nodes as well as the newly constructed union-nodes are only added as successor of a node if it has no other successor that simulates the new addition already. 

The algorithm is shown in \Cref{alg:union_stree}. We highlight that here, in contrast to~\cite{DBLP:journals/sttt/DelzannoRB04}, we are trying to get a simulation-minimal result whereas the authors of that work just run simulation minimization afterwards.

\begin{algorithm}
\caption{Union($n_S$, $n_T$, $layer$)}
\label{alg:union_stree}
\begin{algorithmic}[1]
\State $newNode \gets $ node with value of $n_s$
\If{$layer < k$}
    \State $s_s, s_t \gets 0$
    \While{$s_S < |n_S.successors| \lor s_T < |n_T.successors|$}
        \If{$s_S = |n_S.successors| \lor n_S.successors[s_S].val < n_T.successors[s_T].val$} 
            \State $newNode.addSuccIfNotSimulated(n_T.successors[s_T])$
            \State $s_T++$
        \ElsIf{$s_t = |n_t.successors| \lor n_S.successors[s_S].val > n_T.successors[s_T].val$} 
            \State $newNode.addSuccIfNotSimulated(n_S.successors[s_S])$
            \State $s_S++$
        \Else    
            \State $newSucc \gets Union(n_S.successors[s_S], n_T.successors[s_T], layer + 1)$
            \If{$newSucc \neq null$}
                \State $newNode.addSucc(newChild)$
            \EndIf
            \State $s_s++$
            \State $s_t++$
        \EndIf
    \EndWhile
\EndIf
\State return $layer.addNodeIfNotSimulated(newNode)$
\end{algorithmic}
\end{algorithm}
\color{black}

\subsection{The intersection operation}

Let $S$ and $T$ once again be the sharing trees representing downsets. The basic approach to computing the intersection $S \cap T$ directly using the sharing trees is similar to that for the union: The result is computed recursively per node-pair of one node from $S$ and one from $T$ by parallelly iterating over the successors of both nodes. In contrast to union however, we additionally adapt an idea by Delzanno et al~\cite{DBLP:journals/sttt/DelzannoRB04}. They describe a way of computing the intersection for upward-closed sets that relies on creating a pre-sharing tree with nodes for all possible combinations of nodes of the two trees in the same layer and subsequent reduction by enforcing the sharing tree conditions. 

In our algorithm, see \Cref{alg:intersect_st}, we create the new nodes for every pair of one node from $S$ and one from $T$ in the same layer and set its value to the minimum of the two. The successors are then created recursively with every combination of the successors in $S$ and $T$. Just as for the union, the nodes are only added to the final result if they are not simulated by another node already present. 

However, this check has to be implemented differently due to the order in which new successors are added. For the union, the order in which successors are iterated means every successor that is added will have a value smaller than all previously added successors. Consequently, it is only necessary to check simulation in the direction \(existingSuccessor \xrightarrow{F} newSuccessor\). In intersection on the other hand, the order in which successors are added is not guaranteed. Therefore, a newly added successor might have an equal or even larger value than the existing ones. The check \(existingSuccessor \xrightarrow{F} newSuccessor\) will still be executed in all cases, but if it is determined that the new successor should still be added, additional steps are taken.

We will address the two possible cases separately. If a new successor has to be added that has the same value as an existing one, we follow the approach presented by Zampuniéris~\cite{zampunieris1997sharing}: By definition, a node may not have two successors with the same value, therefore we replace the existing successor with the result of the union of the current successor and the new successor to be added. If the new successor has a value larger than the last successor, we add it in the corresponding position in the successor list, following the descending order. Now it is necessary to also check the simulation in the direction \(newSuccessor \xrightarrow{F} existingSuccessor\) for all successors after the insertion point, \ie those with a smaller value. If it is found that the new successor simulates any existing one, it is removed.

\begin{algorithm}
\caption{ST Intersection}
\label{alg:intersect_st}
\begin{algorithmic}[1]
\State $n_s, n_t \gets$ ST nodes
\State $layer \gets$ layer of the two nodes
\State $father \gets$ ST node
\Function{intersect}{$n_S$, $n_T$, $layer$, $father$}
    \State $newNode \gets $ node with value $min(n_S.val, n_T.val)$
    \If{$layer < d$}
        \For{$s_S = 0; s_S < |n_S.successors|; s_S++$}
            \For{$s_T = 0; s_T < |n_t.successors|; s_T++$}
                \State $intersectResult \gets intersect(n_s.successors[s_S], n_T.successors[s_t],$\\
                \hspace{5cm}$layer + 1, newNode)$
                \If{$intersectResult \neq null$}
                    \State $newNode.addSucc(intersectResult)$
                \EndIf
            \EndFor
        \EndFor
        \If{$|newNode.successors| = 0$}
            \State return $null$
        \EndIf
    \EndIf
    \If{$father \neq null$} 
        \State return $layer.addNodeIfNotSimulated(newNode, father)$
    \Else
        \State return $layer.addNode(newNode)$
    \EndIf
\EndFunction
\end{algorithmic}
\end{algorithm}

\section{Strategy synthesis for parity games}
To obtain a strategy for the even player from $\cpre^*$ we follow~\cite{DBLP:conf/stacs/Jurdzinski00,DBLP:journals/ita/BernetJW02} and, from each vertex $\vec{u}$ controlled by the even player, choose a successor $v$ so as to minimize the minimal element in $\{\vec{c} \in \cpre^*(v) \mid \vec{0} \leq \vec{c}\}$ according to the co-lexicographic order where all dimensions $1 \leq i \leq d$ such that $2i < p(u)$ are ignored. 

\newpage 

\section{LTL-realizability on instances with large sets}\label{sec:5percent}

\begin{figure}[h!]
  \centering
  \begin{subfigure}[t]{0.5\textwidth}
    \includegraphics[height=5cm]{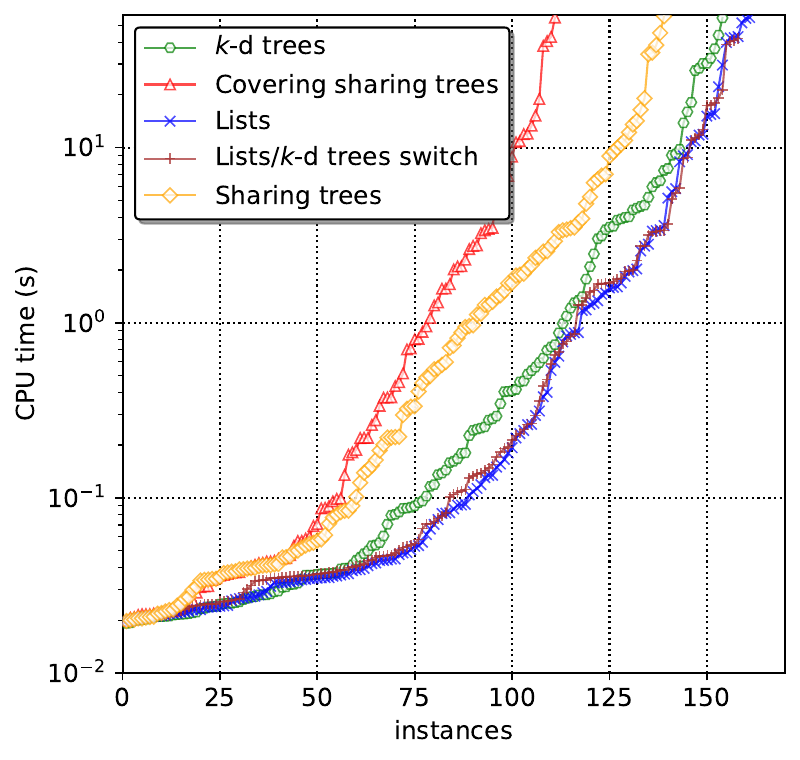}
    \caption{Benchmarking LTL-realizability.}
  \end{subfigure}%
  \begin{subfigure}[t]{0.5\textwidth}
    \includegraphics[height=5cm]{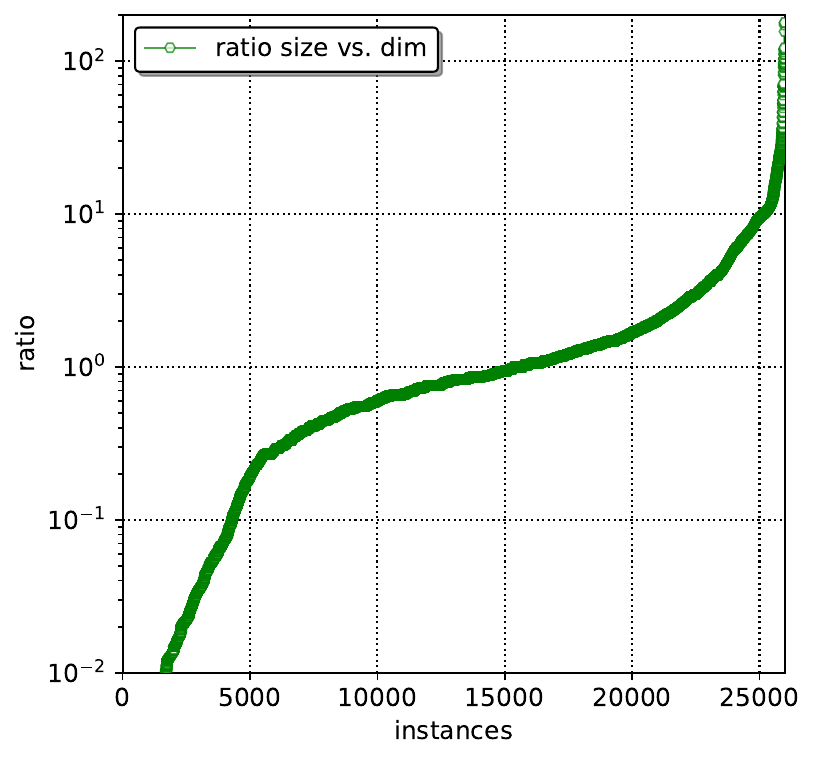}
    \caption{Ratio of size vs.\ dimension.}
  \end{subfigure}
  \caption{Survival plots of downset-based LTL-realizability focusing on the
    test cases that have at least one set that is in the top 5\% in terms of
    ratio size vs.\ dimension.}
  \label{fig:ab2}
\end{figure}

\section{Progress in \kdtree performances based on implementations}\label{sec:old-kdtrees}

The following graphic reports on the gains in performance obtained by building the \kdtree as described in \Cref{sec:build-kdtree} (instead of the classic presorting proposed in computational geometry books) and optimizing
the memory management of \kdtrees, which are the main difference between the
implementation provided in~\cite{DBLP:conf/tacas/CadilhacP23}, and our
implementation. This figure is a survival plot with the same parameters as the
ones of \Cref{sec:ab}.

\begin{figure}[h!]
  \centering
  \includegraphics[height=5cm]{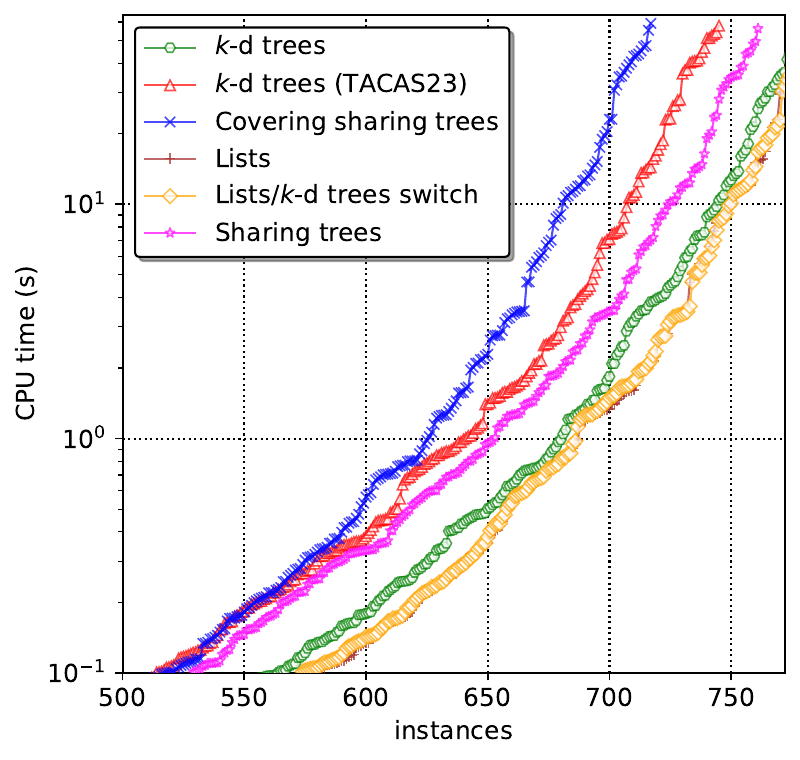}
  \caption{Survival plot of \kdtree-based parity-game solving.}
  \label{fig:progress-kdtrees}
\end{figure}

\section{Memory usage of different data structures in parity-game
  solving}\label{sec:mem}

The following graphic depicts the memory usage of each data structure on all the
parity games benchmarked.  Each vertical line corresponds to one instance, and
these lines are ordered in such a way that the list implementation never
decreases in usage.

\begin{figure}[h!]
  \centering
  \includegraphics[height=5cm]{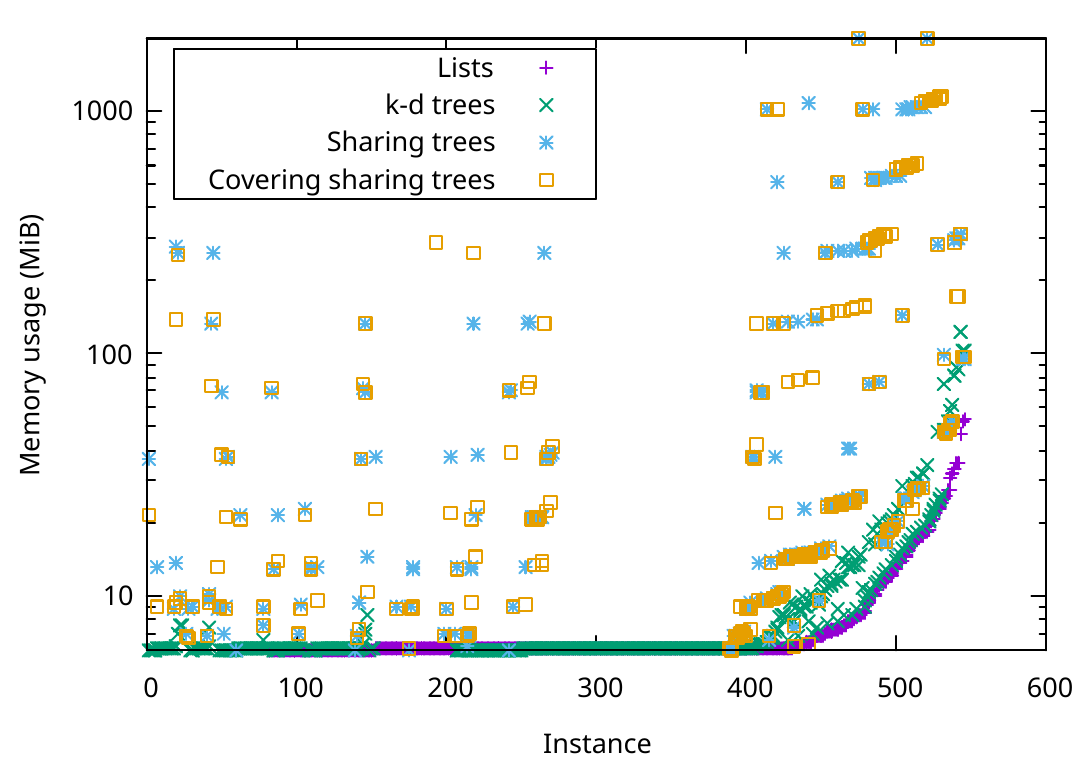}
  \caption{Memory usage of data structures on parity games}
  \label{fig:mem}
\end{figure}

\section{Random benchmarks}

We evaluate the performances of downset operations on random samples.  For
membership, we create a random antichain of size \(t\), and query \(2t\) elements
on the data structures, half of which belong to the downset.  For union and
intersection, we create a random antichain of size \(t\), another antichain of
size \(t\) that overlaps the first for \emph{half} the elements, and perform the
operation.  The dimension is set to 32,000, so that the asymptotic behaviors can
be observed. 

\begin{center}
    \begin{tikzpicture}[scale=0.6,transform shape]
  \datavisualization [our system,xshift=-7cm,
  visualize as smooth line/.list={kdtree,vector,ratio},
  /pgf/data/evaluator=\pgfmathparseFPU,
  style sheet=strong colors,
  style sheet=vary dashing,
  ]

    data[set=kdtree,separator=\space] {
      x y
      00010 0.00139017
00200 0.0376065
00500 0.10973
00800 0.197025
01000 0.258616
01500 0.453138
02000 0.659693
02500 0.92903
03000 1.20033
04000 1.83718
05000 2.61694
06000 3.40566
08000 5.13017
10000 7.3961
    }
    data [set=vector,separator=\space] {
      x y
      00010 0.000272077
00200 0.0169045
00500 0.0646421
00800 0.127349
01000 0.194553
01500 0.385531
02000 0.698927
02500 1.03962
03000 1.48862
04000 2.57431
05000 3.96263
06000 5.49864
08000 9.58891
10000 14.9721
}
data [set=ratio,separator=\space] {
  x r
 00010 0.19571491256465037
00200 0.44951005810165795
00500 0.58910143078465327
00800 0.64635959903565532
01000 0.75228524143904474
01500 0.85080262524882055
02000 1.0594731185566619
02500 1.1190381365510262
03000 1.2401756183716146
04000 1.401229057577374
05000 1.5142227181364494
06000 1.6145592924719436
08000 1.8691212961753705
10000 2.0243236300212275
}

info{
    \node[draw,anchor=west,drop shadow,fill=white] at ($(data visualization
    bounding box.north west) + (0.2cm,-0.5cm)$) {membership};
    };

  \datavisualization [our system,
  visualize as smooth line/.list={kdtree,vector,ratio},
  /pgf/data/evaluator=\pgfmathparseFPU,
  style sheet=strong colors,
  style sheet=vary dashing,
  every label in legend/.style={node style={font=\footnotesize}},
  legend=south outside,
  kdtree={label in legend={text=\kdtree}},
  vector={label in legend={text=lists}},
  ratio={label in legend={text={ratio\ \ \begin{tabular}{c}lists\\\hline \kdtree\end{tabular}}}}
  ]
    data[set=kdtree,separator=\space] {
      x y
      00010 0.00940887
00200 0.0972032
00500 0.252272
00800 0.423089
01000 0.545747
01500 0.868349
02000 1.19528
02500 1.5647
03000 1.92452
04000 2.70064
05000 3.55116
06000 4.48819
08000 6.41159
10000 8.63419
    }
    data [set=vector,separator=\space] {
      x y
      00010 0.00487315
00200 0.055265
00500 0.132979
00800 0.260137
01000 0.348814
01500 0.696163
02000 1.05058
02500 1.6606
03000 2.15851
04000 3.57888
05000 5.45716
06000 7.47485
08000 12.8565
10000 20.1002
}
data [set=ratio,separator=\space] {
  x r
 00010 0.51793148380198684
00200 0.56855124111140376
00500 0.52712548360499789
00800 0.61485172150540435
01000 0.63914964259995932
01500 0.80170875995711399
02000 0.87894049929723594
02500 1.061289704096632
03000 1.1215835636938043
04000 1.3251969903430298
05000 1.5367260275515606
06000 1.6654486552485521
08000 2.0051968388496455
10000 2.3279774941251006
}
info{
    \node[draw,anchor=west,drop shadow,fill=white] at ($(data visualization
    bounding box.north west) + (0.2cm,-0.5cm)$) {union};
    };
  \datavisualization [our system,xshift=7cm,
  visualize as smooth line/.list={kdtree,vector,ratio},
  /pgf/data/evaluator=\pgfmathparseFPU,
  style sheet=strong colors,
  style sheet=vary dashing]

    data[set=kdtree,separator=\space] {
      x y
00003 0.00741401
00014 0.0622973
00022 0.120903
00028 0.202634
00031 0.222358
00038 0.329004
00044 0.408013
00050 0.649798
00054 0.732997
00063 0.940415
00070 1.3967
00077 1.573
00089 2.00309
00100 2.9501
    }
    data [set=vector,separator=\space] {
      x y
00003 0.000990911
00014 0.0497804
00022 0.112236
00028 0.211731
00031 0.293809
00038 0.639374
00044 1.02485
00050 1.59751
00054 1.96174
00063 3.00242
00070 4.23863
00077 5.46858
00089 8.2763
00100 12.1454
}
data [set=ratio,separator=\space] {
  x r
00003 0.1336538526384507
00014 0.79907796967123779
00022 0.92831443388501533
00028 1.0448937493214367
00031 1.3213331654359186
00038 1.9433623907308117
00044 2.5118072218287164
00050 2.4584717096697744
00054 2.6763274610946568
00063 3.1926543068751561
00070 3.0347461874418267
00077 3.4765289256198351
00089 4.1317664208797416
00100 4.1169451882987023
}
info{
    \node[draw,anchor=west,drop shadow,fill=white] at ($(data visualization
    bounding box.north west) + (0.2cm,-0.5cm)$) {intersection};
    };
\end{tikzpicture}
\end{center}

\noindent
For this experiment, our goal was to determine when \kdtrees start performing better than list-based antichains. Hence, to avoid clutter, we focused on these two data structures. Nevertheless, we did try both sharing and covering sharing trees on a number of large-dimension examples (> 8,000) and remark that it is much slower than the other data structures---presumably because, despite sharing, the resulting automaton is rather ``deep and narrow''.

\end{document}